\documentclass{article}
\usepackage{amsmath,amsthm,amsfonts,braket,algorithm,graphicx}
\usepackage{fullpage}

\theoremstyle{definition} 
\theoremstyle{definition} 
\newtheorem {theorem} {Theorem}

\newtheorem {lemma} {Lemma}

\newcommand{\ketbra}[1]{\ket{#1}\bra{#1}}
\newcommand{\kb}[1]{\ket{#1}\bra{#1}}
\newcommand{\bk}[2]{\braket{e_{#1}|e_{#2}}}

\newcommand{\ba}{\bar{\alpha}}
\newcommand{\bb}{\bar{\beta}}

\newcommand{\re}[2]{\mathcal{R}_{#1,#2}}
\newcommand{\im}[2]{\mathcal{I}_{#1,#2}}

\newcommand{\ree}[2]{\widehat{\mathcal{R}}_{#1,#2}}
\newcommand{\ime}[2]{\widehat{\mathcal{I}}_{#1,#2}}

\floatname{algorithm}{Protocol}

\title{Quantum Key Distribution with Mismatched Measurements over Arbitrary Channels}

\author{Walter O. Krawec\\\small{Iona College}\\\small{New Rochelle, NY 10801 USA}\\\small{\texttt{walter.krawec@gmail.com}}}

\begin{document}
\maketitle

\begin{abstract}
In this paper, we derive key-rate expressions for different quantum key distribution protocols.  Our key-rate equations utilize multiple channel statistics, including those gathered from mismatched measurement bases - i.e., when Alice and Bob choose incompatible bases.  In particular, we will consider an Extended B92 and a two-way semi-quantum protocol.  For both these protocols, we demonstrate that their tolerance to noise is higher than previously thought - in fact, we will show the semi-quantum protocol can actually tolerate the same noise level as the fully quantum BB84 protocol.  Along the way, we will also consider an optimal QKD protocol for various quantum channels.  Finally, all the key-rate expressions which we derive in this paper are applicable to any arbitrary, not necessarily symmetric, quantum channel.
\end{abstract}

\section{Introduction}
Quantum key distribution protocols allow two users, Alice ($A$) and Bob ($B$), to establish a shared secret key secure against an all-powerful adversary Eve ($E$) who is bounded only by the laws of physics - an end unattainable through classical means alone.  Several such protocols have been developed since the original BB84 \cite{QKD-BB84} (the reader is referred to \cite{QKD-survey} for a general survey) and many of them include rigorous proofs of unconditional security.  Such a proof of security generally involves determining a bound on the protocol's key-rate (to be defined shortly, though roughly speaking, it is the ratio of secret key bits to qubits sent) as a function of the observed noise in the quantum channel.

In this paper, we consider several QKD protocols, and derive key-rate expressions based on multiple channel statistics.  Furthermore, our key-rate bounds will utilize statistics from mismatched measurements; that is to say, those measurement outcomes where A and B's choice of bases are incompatible - events which are typically discarded by the protocol specification (there are exceptions as we mention next section).  In fact, by using these mismatched measurement results, the key-rate bounds we derive demonstrate that many of the protocols we consider here can actually tolerate higher levels of noise than previously thought.

Following an overview of related work, we will first derive a general approach to deriving key-rate expressions, in the asymptotic scenario, for a wide-range of discrete variable QKD protocol utilizing all possible channel statistics, including mismatched measurement results.  Secondly, we apply this technique to two, very different protocols (an Extended B92 \cite{QKD-B92-extended}, and a two-way semi-quantum protocol \cite{SQKD-first}), deriving new key-rate expressions applicable to arbitrary, possibly asymmetric quantum channels, and, for both protocols, resulting in new bounds which are substantial improvements over previous work with these protocols; in particular, our new key rate expression for the Extended B92 protocol and the semi-quantum protocol will show that these protocols can withstand higher levels of noise than previously thought.  Along the way, we will also use our method to investigate optimal QKD protocols for asymmetric channels.

\subsection{Related Work}

We are not the first to consider the use of mismatched measurement outcomes for quantum key distribution.  Nor are we the first to show that these statistics can lead to improved key rates.  Indeed, in the 1990's Barnett et al. \cite{QKD-Tom-First} showed that mismatched measurement results may be used to better detect an eavesdropper using an intercept-and-resend attack.

In \cite{QKD-Tom-KeyRateIncrease}, mismatched measurement bases were applied to the four-state and six-state BB84 protocols \cite{QKD-BB84}.  This method was shown to improve the key rate for certain quantum channels, namely the amplitude damping channel and rotation channel.  They also derived expressions for non symmetric channels.  In \cite{QKD-Tom-KeyRateMismatchedDistill}, mismatched measurement results were actually used to distill a raw key (as opposed to being used only for channel tomography) - a modified BB84 protocol was adopted and this method was shown to improve the key rate for certain channels.

In \cite{QKD-Tom-BB84NarrowAngle}, a modified two-basis BB84 was developed where the first basis was the standard computational $Z$ basis ($\{\ket{0}, \ket{1}\}$), while the second consisted of states $\ket{0}_\theta = \cos\theta\ket{0} + \sin\theta\ket{1}$ and $\ket{1}_\theta = -\sin\theta\ket{0} + \cos\theta\ket{1}$ where $\theta < \pi/4$.  The authors of that work showed that for small $\theta$, mismatched measurement bases can still be used to gain good channel estimates while also allowing $A$ and $B$ to use mismatched measurement bases to distill their key (since, for $\theta$ small, even with differing bases, their measurement results will be nearly correlated).

Mismatched measurements were used in \cite{SQKD-Krawec-ReflectionSecurity} in order to get better channel statistics for a single-state semi-quantum protocol first introduced in \cite{SQKD-Single-Security}.  Though single-state semi-quantum protocols utilize two-way quantum channels, they admit many simplifications which ease their security analysis.  In this paper, we consider a multi-state semi-quantum protocol (which are more difficult to analyze) and show mismatched measurements improve its key-rate; indeed, our new key rate bound derived in this paper shows this semi-quantum protocol has the same noise tolerance as the fully-quantum BB84 protocol.

In \cite{QKD-Tom-threestate1}, it was proven, using mismatched measurement bases, that the three-state BB84 protocol from \cite{QKD-BB84-three-state,QKD-BB84-three-state-v2} has a key rate equal to that of the full four state BB84 protocol assuming a symmetric attack.  Also a four-state protocol using three bases has a key rate equal to that of the full six-state BB84 protocol.

In this paper, building off of our conference paper in \cite{QKD-Tom-threestate-Krawec} (where we only considered three states for parameter estimation), we will apply mismatched measurements to non-BB84 style protocols and to protocols relying on two-way quantum channels.  After an introduction to our notation, we will first explain the parameter estimation method and our technique.  We will then apply it to the Extended B92 protocol \cite{QKD-B92-extended} and derive an improved key-rate bound for it.  We will then use our method to consider an ``optimal'' QKD protocol.  Finally, we will analyze a multi-state semi-quantum protocol from \cite{SQKD-first} which relies on a two-way quantum channel.  This new proof of security will derive a far more optimistic bound on the key rate expression than the one previously constructed in \cite{SQKD-Krawec-SecurityProof} (the latter did not use mismatched measurement bases).

\subsection{Notation}

We now introduce some notation we will use.  Let $H(\cdot)$ be the Shannon entropy function, namely:
\[
H(p_1, p_2, \cdots, p_n) = -\sum_{i=1}^np_i\log p_i,
\]
where all logarithms in this paper are base two unless otherwise specified.  We will occasionally use the notation $H(\{p_i\}_{i=1}^n)$ to mean $H(p_1, \cdots, p_n)$.  We denote by $h(x)$ the binary Shannon entropy function, namely: $h(x) = -x\log x - (1-x)\log (1-x)$.

We write $S(\rho)$ to mean the von Neumann entropy of the density operator $\rho$.  If $\rho$ is finite dimensional (and all systems in this paper are finite dimensional), then let $\{\lambda_1, \cdots, \lambda_n\}$ be its eigenvalues.  In this case $S(\rho) = -\sum_i\lambda_i\log \lambda_i$.

If $\rho$ acts on a bipartite Hilbert space $\mathcal{H}_A\otimes \mathcal{H}_B$ we will often write $\rho_{AB}$.  When we write $\rho_A$ we mean the result of tracing out $B$'s portion of $\rho_{AB}$ (that is, $\rho_A = tr_B\rho_{AB}$).  Similarly for $\rho_B$ and for systems with three or more subspaces.

Given density operator $\rho_{AB}$ we will write $S(AB)_\rho$ to mean $S(\rho_{AB})$ and $S(B)_\rho$ to mean $S(\rho_B)$.  We denote by $S(A|B)_\rho$ to be the conditional von Neumman entropy defined as $S(A|B)_\rho = S(AB)_\rho - S(B)_\rho = S(\rho_{AB}) - S(\rho_B)$.  If the context is clear, we will forgo writing the ``$\rho$'' subscript.

When we talk about qubits, we will often refer to the $Z$, $X$, and $Y$ bases, the states of which we denote: $Z = \{\ket{0}, \ket{1}\}$, $X = \{\ket{+}, \ket{-}\}$, and $Y = \{\ket{0_Y}, \ket{1_Y}\}$, where:
\begin{align*}
\ket{\pm} &= \frac{1}{\sqrt{2}}(\ket{0} \pm \ket{1})\\
\ket{j_Y} &= \frac{1}{\sqrt{2}}(\ket{0} + (-1)^ji\ket{1}).
\end{align*}

\section{Channel Tomography}\label{section:PE}

We will now describe the parameter estimation method and what may be gleaned from it using our notation.

Let $U$ be a unitary operator acting on the finite dimensional Hilbert space $\mathcal{H}_T\otimes\mathcal{H}_E$ where $\dim\mathcal{H}_T=2$ and $\dim\mathcal{H}_E < \infty$ ($U$ will model Eve's attack operation); $\mathcal{H}_T$ models the qubit ``transit'' space, while $\mathcal{H}_E$ will model the adversary $E$'s private quantum memory.  Without loss of generality, we may describe $U$'s action on states of the form $\ket{i,\chi}_{TE}$, where $\ket{\chi}_E$ is some arbitrary, normalized state in $\mathcal{H}_E$, as follows:
\begin{align}
U\ket{0,\chi} &= \ket{0,e_0} + \ket{1,e_1}\label{eq:U-states}\\
U\ket{1,\chi} &= \ket{0, e_2} + \ket{1,e_3},\notag
\end{align}
where the $\ket{e_i}$ are states in $\mathcal{H}_E$ which are not necessarily normalized nor orthogonal.  Unitarity of $U$ imposes certain obvious restrictions on these states which will become important momentarily.

Let $Z = \{\ket{0}, \ket{1}\}$ be the computational $Z$ basis.  Let $\alpha, \beta \in [0,1]$ and denote by $\ba = \sqrt{1-\alpha^2}$ and $\bb = \sqrt{1-\beta^2}$.  We will always assume the choice of $\alpha$ and $\beta$ is public knowledge and, once chosen, remains fixed.

Define $\mathcal{A}_{\alpha} = \{\ket{a}, \ket{\bar{a}}\}$ and $\mathcal{B}_\beta = \{\ket{b}, \ket{\bar{b}}\}$, where:
\begin{align}
\ket{a} &= \alpha\ket{0} + \ba\ket{1} = \alpha\ket{0} + \sqrt{1-\alpha^2}\ket{1}\\
\ket{\bar{a}} &= \ba\ket{0} - \alpha\ket{1} = \sqrt{1-\alpha^2}\ket{0} - \alpha\ket{1}\\\notag\\
\ket{b} &= \beta\ket{0} + \bb i\ket{1} = \beta\ket{0} + i\sqrt{1-\beta}\ket{1}\\
\ket{\bar{b}} &= \bb\ket{0} - \beta i\ket{1} = \sqrt{1-\beta^2}\ket{0} - \beta i \ket{1}.
\end{align}
Clearly $\mathcal{A}_\alpha$ and $\mathcal{B}_\beta$ are both orthonormal bases of $\mathcal{H}_T$.  Note that when $\alpha = \beta = 1/\sqrt{2}$ we have $\mathcal{A} = X$, the Hadamard $X$ basis, and $\mathcal{B} = Y$, the $Y$ basis.  The $Z$ and $X$ bases are customarily used in the BB84 protocol, while the $Z$, $X$, and $Y$ bases are used in the six-state BB84 protocol \cite{QKD-BB84,QKD-renner-keyrate}.  Note also that when $\alpha = \beta = 1$, we have $Z = \mathcal{A} = \mathcal{B}$.

As it will be important later, we note that:
\begin{align}
U\ket{a,\chi} &= \ket{a,f_0} + \ket{\bar{a},f_1}\label{eq:U-states-f}
\end{align}
where, due to linearity of $U$, we have:
\begin{align}
\ket{f_0} &= \alpha^2\ket{e_0} + \alpha\ba\ket{e_2} + \alpha\ba\ket{e_1} + \ba^2\ket{e_3}\label{eq:U-states-f2}\\
\ket{f_1} &= \ba\alpha\ket{e_0} + \ba^2\ket{e_2} - \alpha^2\ket{e_1} - \alpha\ba\ket{e_3}.\notag
\end{align}

As the operator $U$ will be used to model $E$'s attack, we are interested in determining what $A$ and $B$ can learn about it (in particular, the $\ket{e_i}$ states) after performing the parameter estimation protocol described in Protocol \ref{alg:PE}.

\begin{algorithm}
\caption{Parameter Estimation with Mismatched Bases}\label{alg:PE}
Let $\alpha, \beta$ be fixed and public knowledge.  Let $\Psi \subset Z \cup \mathcal{A}_\alpha \cup \mathcal{B}_\beta$ also be fixed and public knowledge.  Finally fix an attack operator $U$ known only to $E$.

Repeat the following procedure for $i = 1, 2, \cdots, M$:
\begin{enumerate}
  \item $A$ chooses a random state $\ket{\psi_i} \in \Psi$ and sends a qubit prepared in this state to $B$.
  \item $E$ captures the qubit before it arrives at $B$'s lab and probes it using her attack operator $U$.  We assume this operator acts on the qubit and a freshly prepared ancilla $\mathcal{H}_E$ prepared in a state $\ket{\chi}_E$ known only to $E$, but constant each iteration.  That is, the system evolves from $\ket{\psi_i,\chi}$ to $U\ket{\psi_i,\chi}$.  The qubit (in $\mathcal{H}_T$) is sent to $B$.
  \item $B$ chooses a basis $Z$, $\mathcal{A}_\alpha$, or $\mathcal{B}_\beta$ (assuming a state from one of these bases appears in $\Psi$) and measures the qubit in that basis.
  \item $A$ and $B$ disclose their choice of basis.  $A$ discloses the state she prepared and $B$ discloses his measurement result.  This communication is done via an authenticated classical channel.  Note that they do not discard any measurement results (i.e., they do not discard iterations where they chose different bases).
\end{enumerate}
\end{algorithm}

We will consider two cases for the choice of $\Psi$ in this paper (others may easily follow from our computations):
\begin{align}
\Psi_3 &= \Psi_3^\alpha = \{\ket{0}, \ket{1}, \ket{a}\}\label{eq:psi3}\\
\Psi_4 &= \Psi_4^{\alpha,\beta} = \{\ket{0},\ket{1}, \ket{a}, \ket{b}\}.\label{eq:psi4}
\end{align}
That is, in the first case, $A$ is limited to only sending three possible states from the four available in the $Z$ and $\mathcal{A}$ bases; in the second she can only send four possible states from the six available in the $Z$, $\mathcal{A}$, and $\mathcal{B}$ bases.  Note that, on step 3 of the parameter estimation protocol above, $B$ will measure in either the $Z$ or $\mathcal{A}$ bases in the first case; in the second case he will measure in the $Z$, $\mathcal{A}$, or $\mathcal{B}$ bases (his outcomes could be any of the four states or six states respectively).  Thus, $A$ and $B$ can, for instance, estimate the probability of $U$ flipping a $\ket{a}$ to a $\ket{\bar{a}}$ but they cannot directly observe the probability of $U$ flipping a $\ket{\bar{a}}$ to a $\ket{a}$ (as $A$ can never prepare the state $\ket{\bar{a}}$).

Note that in this paper, we assume $M$ may be set large enough so as to attain arbitrarily high levels of accuracy in all estimates.  This is clearly without loss of generality in the asymptotic scenario in QKD security; it remains an open question to analyze this scenario in the finite key setting where imprecisions will affect the final key rate computation.

Clearly $A$ and $B$ may estimate $\braket{e_i|e_i}$ by using those iterations where $\ket{\psi_i} \in Z$ and $B$ measures in the $Z$ bases.  Now, denote by $p_{i,j}$, where $\ket{i} \in \Psi$ and $\ket{j} \in Z\cup\mathcal{A}\cup\mathcal{B}$, to be the probability that $B$ measures $\ket{j}$ if $A$ initially sends $\ket{i}$ in step 1 of the protocol, conditioning on the event both parties choose the correct basis for such an outcome.  Also, let $\re{i}{j}$ and $\im{i}{j}$ denote the real part of $\braket{e_i|e_j}$ and the imaginary part of $\braket{e_i|e_j}$ respectively.

Now, consider $p_{0,a}$.  If $A$ sends $\ket{0}$, after evolution by $U$, the state becomes:
\[
\ket{0} \mapsto \ket{0,e_0} + \ket{1,e_1} = \ket{a}(\alpha\ket{e_0} + \ba\ket{e_1}) + \ket{\bar{a}}(\ba\ket{e_0} - \alpha\ket{e_1}),
\]
from which it is clear that $p_{0,a} = \alpha^2\bk{0}{0} + \ba^2\bk{1}{1} + 2\alpha\ba\re{0}{1}$.  Similarly, the following expressions are easily derived:

\begin{align}
p_{0,a} &= \alpha^2\bk{0}{0} + \bar{\alpha}^2\bk{1}{1} + 2\alpha\bar{\alpha}\re{0}{1}\label{eq:p0a}\\
&\Rightarrow \re{0}{1} = \frac{p_{0,a} - \alpha^2\bk{0}{0} - \bar{\alpha}^2\bk{1}{1}}{2\alpha\bar{\alpha}}\notag\\\notag\\
p_{1,a} &= \alpha^2\bk{2}{2} + \bar{\alpha}^2\bk{3}{3} + 2\alpha\bar{\alpha}\re{2}{3}\label{eq:p1a}\\
&\Rightarrow \re{2}{3} = \frac{p_{1,a} - \alpha^2\bk{2}{2} - \bar{\alpha}^2\bk{3}{3}}{2\alpha\bar{\alpha}}\notag\\\notag\\
p_{a,0} &=\alpha^2\bk{0}{0} + \bar{\alpha}^2\bk{2}{2} + 2\alpha\bar{\alpha}\re{0}{2}\label{eq:pa0}\\
&\Rightarrow \re{0}{2} = \frac{p_{a,0} - \alpha^2\bk{0}{0} - \bar{\alpha}^2\bk{2}{2}}{2\alpha\bar{\alpha}}\notag\\\notag\\
p_{a,\bar{a}} &=\alpha^2\bar{\alpha}^2(\bk{0}{0} + \bk{3}{3}) + \bar{\alpha}^4\bk{2}{2} + \alpha^4\bk{1}{1}\label{eq:paa}\\
&+ 2\bar{\alpha}^3\alpha \re{0}{2} - 2\bar{\alpha}\alpha^3 \re{0}{1} - 2\alpha^2\bar{\alpha}^2 \re{0}{3}\notag\\
&- 2\alpha^2\bar{\alpha}^2 \re{1}{2} - 2\alpha\bar{\alpha}^3 \re{2}{3} + 2\alpha^3\bar{\alpha}\re{1}{3}\notag
\end{align}

\begin{align}
p_{0,b} &= \beta^2\bk{0}{0} + \bar{\beta}^2\bk{1}{1} + 2\beta\bar{\beta}\im{0}{1}\label{eq:p0b}\\
&\Rightarrow \im{0}{1} = \frac{p_{0,b} - \beta^2\bk{0}{0} - \bar{\beta}^2\bk{1}{1}}{2\beta\bar{\beta}}\notag\\\notag\\
p_{1,b} &= \beta^2\bk{2}{2} + \bar{\beta}^2\bk{3}{3} + 2\beta\bar{\beta}\im{2}{3}\label{eq:p1b}\\
&\Rightarrow \im{2}{3} = \frac{p_{1,b} - \beta^2\bk{2}{2} - \bar{\beta}^2\bk{3}{3}}{2\beta\bar{\beta}}\notag\\\notag\\
p_{b,0} &= \beta^2\bk{0}{0} + \bar{\beta}^2\bk{2}{2} - 2\beta\bar{\beta}\im{0}{2}\label{eq:pb0}\\
&\Rightarrow \im{0}{2} = \frac{\beta^2\bk{0}{0} + \bar{\beta}^2\bk{2}{2} - p_{b,0}}{2\beta\bar{\beta}}\notag\\\notag\\
p_{b,\bar{b}} &= \beta^2\bar{\beta}^2(\bk{0}{0} + \bk{3}{3}) + \bar{\beta}^4\bk{2}{2} + \beta^4\bk{1}{1}\\
&- 2\bar{\beta}^3\beta \im{0}{2} - 2\bar{\beta}\beta^3 \im{0}{1} - 2\beta^2\bar{\beta}^2 \re{0}{3}\notag\\
&+ 2\beta^2\bar{\beta}^2 \re{1}{2} - 2\beta\bar{\beta}^3 \im{2}{3} - 2\beta^3\bar{\beta}\im{1}{3}\notag
\end{align}

As already mentioned, if $Z \subset \Psi$, then each $\bk{i}{i}$ can be estimated by $A$ and $B$.  If $\Psi = \Psi_3^\alpha$, then $p_{0,a}, p_{1,a}, p_{a,0},$ and $p_{a,\bar{a}}$ are statistics that are observable by $A$ and $B$.  If $0 < \alpha < 1$, then this provides $A$ and $B$ with estimates of the quantities $\re{0}{1}$, $\re{2}{3}$, and $\re{0}{2}$.  Since $U$ is unitary and so $\bk{0}{2} = -\bk{1}{3}$, this also provides $A$ and $B$ with an estimate of $\re{1}{3}$ (i.e., $\re{1}{3} = -\re{0}{2}$).  Finally, using these, along with the Cauchy-Schwarz inequality to bound $-\sqrt{\bk{1}{1}\bk{2}{2}} \le \re{1}{2} \le \sqrt{\bk{1}{1}\bk{2}{2}}$ (both $\bk{1}{1}$ and $\bk{2}{2}$ should be small as they represent the $Z$ basis error rate), $A$ and $B$ may bound $\re{0}{3}$ using $p_{a,\bar{a}}$.  We will show a particular example shortly when we consider a symmetric attack.

If $\Psi = \Psi_4^{\alpha,\beta}$, then $A$ and $B$ may estimate all of the above, in addition to $p_{0,b}, p_{1,b}, p_{b,0}$, and $p_{b,\bar{b}}$.  This, combined with the above, allows $A$ and $B$ to estimate: $\bk{0}{1}, \bk{2}{3}, \bk{0}{2}, \bk{1}{3}$ (i.e., they may estimate both the real parts - using $\mathcal{A}$ - and the imaginary parts - using $\mathcal{B}$).

Finally, the quantity $p_{b,\bar{b}}$, together with $p_{a,\bar{a}}$ may be used to compute a better estimate of $\re{1}{2}$ and $\re{0}{3}$.  To illustrate, assume $\alpha = \beta = 1/\sqrt{2}$ (other values may be used, though the algebra becomes tedious).  Then, the following identities are easy to derive:
\begin{align}
p_{a,\bar{a}} + p_{b, \bar{b}} &= 1 - \frac{1}{2}(\re{0}{1} + \im{0}{1} + \re{2}{3} + \im{2}{3}) - \re{0}{3}\notag\\
\Rightarrow \re{0}{3} &= 1 - p_{a,\bar{a}} - p_{b,\bar{b}} - \frac{1}{2}(\re{0}{1} + \im{0}{1} + \re{2}{3} + \im{2}{3}).\label{eq:r03-gen}\\\notag\\
p_{a,\bar{a}} - p_{b,\bar{b}} &= \frac{1}{2}(\im{0}{1} - \re{0}{1} + \im{2}{3} - \re{2}{3}) - \re{1}{2}\notag\\
\Rightarrow \re{1}{2} &= p_{b,\bar{b}} - p_{a,\bar{a}} + \frac{1}{2}(\im{0}{1} - \re{0}{1} + \im{2}{3} - \re{2}{3})\label{eq:r12-gen}
\end{align}
(Note, in the above, we made use of the fact that $\re{0}{2} + \re{1}{3} = 0$ and $\im{0}{2} + \im{1}{3} = 0$ due to unitarity of $U$.)

\subsection{Symmetric Channels}\label{section:sym}
The above analysis applies to any arbitrary channel.  However, to illustrate the use of the above expressions, let us consider a special case that $E$'s attack operator $U$ is symmetric; in particular, $E$'s attack is such that its action on $\mathcal{H}_T$ may be modeled as a depolarization channel with parameter $Q$.  That is, for any two-dimensional density operator $\rho$, we have, after $E$'s attack:
\[
\mathcal{E}_Q(\rho) = (1-2Q)\rho + QI,
\]
where $I$ is the two-dimensional identity operator.  This assumption is common in the QKD literature.  However, we comment that, in the following, we will only utilize statistics that $A$ and $B$ are capable of observing directly; our use of a depolarization channel is simply to demonstrate our technique and to put concrete numbers to these many parameters.  Furthermore, since only observable statistics are used, this depolarization channel assumption is entirely enforceable by $A$ and $B$.

Let us first consider $\Psi_3^\alpha$ (the same set that was considered in our conference paper \cite{QKD-Tom-threestate-Krawec}), with $\alpha \in (0,1)$.  If $A$ sends a $Z$ state $\ket{i}$ for $i \in \{0,1\}$, then the qubit's state when it arrives at $B$'s lab is:
\[
\mathcal{E}_Q(\ketbra{i}) = (1-2Q)\ketbra{i} + QI.
\]
From this, it is clear that $p_{0,1} = p_{1,0} = Q$.  Thus $\bk{0}{0} = \bk{3}{3} = 1-Q$ and $\bk{1}{1} = \bk{2}{2} = Q$.

Now, if $A$ sends $\ket{0}$, it is clear that $p_{0,a} = (1-2Q)\alpha^2 + Q\alpha^2 + Q\bar{\alpha}^2 = (1-Q)\alpha^2 + Q\bar{\alpha}^2$.  From this, along with Equation \ref{eq:p0a}, we conclude $\re{0}{1} = 0$.  (The reader will observe that, when $\alpha = 1/\sqrt{2}$, and so $\ket{a} = \ket{+}$, we have $p_{0,a} = p_{0,+} = 1/2$ as expected.)  Similarly, we conclude that $\re{0}{2} = \re{2}{3} = \re{1}{3} = 0$.

Substituting this into Equation \ref{eq:paa}, and solving for $\re{0}{3}$, we find:
\begin{align}
\re{0}{3} &= \frac{2\alpha^2\bar{\alpha}^2(1-Q) + (\bar{\alpha}^4+\alpha^4)Q - p_{a,\bar{a}} - 2\alpha^2\bar{\alpha}^2\re{1}{2}}{2\alpha^2\bar{\alpha}^2}\notag\\
&=1 - 2Q + \frac{Q - p_{a,\bar{a}}}{2\alpha^2\bar{\alpha}^2} - \re{1}{2},
\end{align}
where we used the fact that $1 = (\alpha^2 + \bar{\alpha}^2)^2 = \alpha^4 + 2\alpha^2\bar{\alpha}^2 + \bar{\alpha}^4$ and so $\alpha^4 + \bar{\alpha}^4 = 1 - 2\alpha^2\bar{\alpha}^2$.

If $Q = p_{a,\bar{a}}$ (which it would be in this depolarization channel scenario), then the above simplifies to:
\[
\re{0}{3} = 1 - 2Q - \re{1}{2},
\]
an expression with no dependence on $\alpha$ (other than $\alpha \ne 0,1$).  By the Cauchy-Schwarz inequality, we have $\re{1}{2} \in [-Q, Q]$ thus implying $\re{0}{3} \in [1-3Q, 1-Q]$.

Now, if $\Psi = \Psi_4^{\alpha,\beta}$, for $\alpha,\beta \in (0,1)$, then, in addition to the above, we also find $\im{0}{1} = \im{2}{3} = \im{0}{2} = \im{1}{3} = 0$ (using a similar method as described above, but with statistics $p_{0,b}$, $p_{1,b}$, and $p_{b,0}$).  By realizing that $p_{a,\bar{a}} = p_{b,\bar{b}}$, we also have $\re{1}{2} = 0$.  Thus, in this setting, it holds that $\re{0}{3} = 1-2Q$.

\subsection{Concerning the Parameter $\alpha$ and $\beta$}

Rather interestingly, at least in the asymptotic scenario, the choice of parameters for $\alpha$ and $\beta$ do not alter the parameter estimation as proven in \cite{QKD-Tom-threestate1,QKD-Tom-BB84NarrowAngle}.  Below, we include a sketch of the argument, using our notation, only for completeness of this paper.

Denote by $\re{i}{j}$ to be, as before, $\re{i}{j} = Re\braket{e_i|e_j}$ (similarly $\im{i}{j}$).  This is the actual value of the inner product of these two states.  Denote by $\ree{i}{j}$ and $\ime{i}{j}$ to be $A$ and $B$'s estimate of these quantities after parameter estimation.  We now determine how the choice of $\alpha$ and $\beta$ affect these estimates, assuming we have perfect parameter estimation of the required $p_{i,j}$ quantities (in the asymptotic scenario, which we focus on in this paper, this is not a problem; we leave the non-perfect case as future study).

Clearly, from Equations \ref{eq:p0a}, \ref{eq:p1a}, and \ref{eq:pa0}, we have $\re{0}{1} = \ree{0}{1}$, $\re{2}{3} = \ree{2}{3}$, $\re{0}{2} = \ree{0}{2}$ and $\re{1}{3} = \ree{1}{3}$.  From Equations \ref{eq:p0b}, \ref{eq:p1b}, and \ref{eq:pb0}, we get the same result for the imaginary parts of these quantities.  The choice of $\alpha$ and $\beta$, therefore, do not affect these estimates.

What remains is to consider how these two parameters affect $\ree{0}{3}$ and $\ree{1}{2}$.  Using Equation \ref{eq:paa}, $A$ and $B$ will estimate these two quantities using the identity:
\begin{align*}
2\alpha^2\bar{\alpha}^2\ree{0}{3} + 2\alpha^2\bar{\alpha}^2\ree{1}{2} &=\alpha^2\bar{\alpha}^2(\bk{0}{0} + \bk{3}{3}) + \bar{\alpha}^4\bk{2}{2} + \alpha^4\bk{1}{1}\\
&+ 2\bar{\alpha}^3\alpha \ree{0}{2} - 2\bar{\alpha}\alpha^3 \ree{0}{1} - 2\alpha\bar{\alpha}^3 \ree{2}{3} + 2\alpha^3\bar{\alpha}\ree{1}{3}\\
&- p_{a,\bar{a}}
\end{align*}

Using the fact that the estimates for all real parts on the right-hand side of the above expression are, indeed, the actual values, the above implies (substituting in $p_{a, \bar{a}}$ for the actual values $\re{i}{j}$):
\begin{align*}
&2\alpha^2\bar{\alpha}^2(\ree{0}{3} + \ree{1}{2}) = 2\alpha^2\bar{\alpha^2}(\re{0}{3} + \re{1}{2})\\
\iff & \ree{0}{3} + \ree{1}{2} = \re{0}{3} + \re{1}{2},
\end{align*}
from which we see there is no dependence on $\alpha$.  Of course this parameter may affect the key rate for certain protocols (as we see later), it does not affect the parameter estimation process assuming perfect estimates of $p_{i,j}$.

When we consider also the $\mathcal{B}$ basis, we additionally get the following identity (using $p_{b,\bar{b}}$):
\[
\ree{0}{3} - \ree{1}{2} = \re{0}{3} - \re{1}{2}.
\]

Combining this with the above, we find $\ree{0}{3} = \re{0}{3}$ and $\ree{1}{2} = \re{1}{2}$.  Thus, in this case, we get exact estimates of the actual values of $E$'s attack which, again, do not depend on $\alpha$ or $\beta$.

The above all assumed perfect accuracy in our estimation of the quantities $p_{i,j}$.  In a realistic scenario, this will not be the case.  In such a setting, the choice of $\alpha$ and $\beta$ will no doubt alter the estimation of the quantities $\re{i}{j}$ and $\im{i}{j}$.  We leave this study as future work.

\section{Applications}

In this section, we will apply the above estimation method to analyze three different QKD protocols.  First, we will analyze the Extended B92 protocol first introduced in \cite{QKD-B92-extended}.  After this, we will consider an ``optimized'' QKD protocol.  Finally, we will apply our technique to a protocol relying on a two-way quantum channel (thus allowing the attacker two opportunities to interact with the qubit each iteration).  For all these protocols, we will derive key-rate expressions as functions of the various $\re{i}{j}$ quantities discussed in the last section from which Protocol \ref{alg:PE} may be used to discover bounds on.

\subsection{QKD Security}

A QKD protocol operates by first performing several iterations of a \emph{quantum communication stage}, the result of which is an $N$-bit \emph{raw key}: a bit-string that is partially secret and partially correlated ($A$ and $B$ each have their own version of the raw key).  Following this, assuming $A$ and $B$ do not abort, they perform an error correction protocol and a privacy amplification protocol (see \cite{QKD-survey}), resulting in a secret key of size $\ell(N)$-bits (possibly $\ell(N) = 0$ if the noise rate is too high - i.e., if $E$ has potentially too much information on the raw key).  The goal is to determine $\ell(N)$ as a function of the noise rate (in our case, as functions of the various $p_{i,j}$ values considered in the previous section).  Actually, we are interested in computing a protocol's \emph{key rate} in the \emph{asymptotic scenario}, defined to be:
\begin{equation}
\text{key-rate} = r = \lim_{N\rightarrow\infty} \frac{\ell(N)}{N}.
\end{equation}

We will first consider collective attacks: these are attacks where $E$ applies the same attack operation each iteration of the protocol.  Thus each iteration is independent of all others.  Later we will consider \emph{general attacks} where no such restrictions are placed on the adversary.

Assuming collective attacks, it was shown in \cite{QKD-winter-keyrate,QKD-renner-keyrate} that the key rate is:
\[
r = \inf(S(A|E) - H(A|B)),
\]
where the infimum is over all collective attacks which induce the observed statistics.  The computation of $H(A|B)$ is trivial given values $p_{i,j}$ for $i,j \in \{0,1\}$.  The challenge is to compute $S(A|E)$.

We will use the following lemma and theorem to compute the conditional von Neumann entropy $S(A|E)$ of the quantum systems which will appear later in our analysis of these various QKD protocols.

\begin{lemma}\label{lemma:entropy-computation}
Let $\mathcal{H} = \mathcal{H}_X\otimes\mathcal{H}_Y$ be a finite dimensional Hilbert space with $\{\ket{1}_X, \ket{2}_X, \cdots, \ket{n}_X\}$ an orthonormal basis of $\mathcal{H}_X$.  Consider the following density operator:
\[
\rho = \sum_{i=1}^n p_i\kb{i}_X \otimes\sigma_Y^{(i)},
\]
where $\sum_i p_i = 1$, $p_i \ge 0$, and each $\sigma_Y^{(i)}$ a Hermitian positive semi-definite operator of unit trace acting on $\mathcal{H}_Y$.  Then the von Neumann entropy of $\rho$ is:
\begin{equation}
S(\rho) = H(p_1, p_2, \cdots, p_n) + \sum_{i=1}^np_iS\left(\sigma_Y^{(i)}\right).
\end{equation}
\end{lemma}
\begin{proof}
The proof of this is straight-forward.  See, for instance, \cite{SQKD-Krawec-SecurityProof}.
\end{proof}

\begin{theorem}\label{thm:entropy}
Let $\mathcal{H}_A\otimes\mathcal{H}_E$ be a finite dimensional Hilbert space.  Consider the following density operator:
\begin{equation}\label{eq:thm:1}
\rho_{AE} = \frac{1}{N}\left( \kb{0}_A \otimes \left[ \sum_{i=1}^M \kb{g_i^0} \right] + \kb{1}_A \otimes \left[\sum_{i=1}^M\kb{g_i^1}\right]\right),
\end{equation}
where $N > 0$ is a normalization term, $M < \infty$ (note, this $M$ has no relation to the $M$ from the parameter estimation procedure described last section), and each $\ket{g_i^j} \in \mathcal{H}_E$ (these are not necessarily normalized, nor orthogonal, states; also it might be that $\ket{g_i^j} \equiv 0$ for some $i$ and $j$).  Let $N_i^j = \braket{g_i^j|g_i^j} \ge 0$.  Then:
\begin{equation}\label{eq:thm:entropy}
S(A|E)_\rho \ge \sum_{i=1}^M \left( \frac{N_i^0+N_i^1}{N} \right) \cdot S_i,
\end{equation}
where:
\begin{equation}\label{eq:thm:cond-entropy}
S_i = \left\{ \begin{array}{ll}
h\left( \frac{N_i^0}{N_i^0+N_i^1} \right) - h\left(\lambda_i\right) & \text{ if } N_i^0 > 0 \text{ and } N_i^1 > 0\\\\
0 & \text{ otherwise}
\end{array}\right.
\end{equation}
and:
\begin{equation}\label{eq:thm:lambda}
\lambda_i = \frac{1}{2} + \frac{ \sqrt{(N_i^0-N_i^1)^2 + 4Re^2\braket{g_i^0|g_i^1}} }{2(N_i^0+N_i^1)}.
\end{equation}
\end{theorem}
\begin{proof}
Note that $N_i^j = 0$ if and only if $\ket{g_i^j} \equiv 0$.  Also note that, since $\rho_{AE}$ is a density operator, and thus unit trace, it must be that $N = \sum_{i,j} N_i^j$.  Assume, first, that all $N_i^j > 0$.  We may write $\rho_{AE}$ in the form:
\begin{equation}\label{eq:thm:rho-exp}
\rho_{AE} = \sum_{i=1}^M \frac{N_i^0+N_i^1}{N} \sigma_{AE}^{(i)},
\end{equation}
where:
\[
\sigma_{AE}^{(i)} = \frac{N_i^0}{N_i^0+N_i^1} \kb{0}_A \otimes \frac{\kb{g_i^0}}{N_i^0} + \frac{N_i^1}{N_i^0+N_i^1} \kb{1}_A \otimes \frac{\kb{g_i^1}}{N_i^1}.
\]

Note that each $\sigma_{AE}^{(i)}$ is a Hermitian positive semi-definite operator of unit trace.

Let $\mathcal{H}_C$ be the $M$-dimensional Hilbert space spanned by orthonormal basis $\{\ket{\chi_1}, \ket{\chi_2}, \cdots, \ket{\chi_M}\}$, and let $\tau_{AEC}$ be the following density operator:
\[
\tau_{AEC} = \sum_{i=1}^M \frac{N_i^0+N_i^1}{N} \kb{\chi_i} \otimes \sigma_{AE}^{(i)}.
\]

Clearly, $\tau_{AE} = tr_C\tau_{AEC} = \rho_{AE}$.  Due to the strong sub additivity of von Neumann entropy, we have:
\[
S(A|E)_\rho = S(A|E)_\tau \ge S(A|EC)_\tau.
\]

By definition of conditional entropy, we have $S(A|EC)_\tau = S(AEC)_\tau - S(EC)_{\tau}$.  Applying Lemma \ref{lemma:entropy-computation} twice, we have:
\begin{align*}
S(AEC)_\tau &= H\left( \left\{ \frac{N_i^0 + N_i^1}{N} \right\}_{i=1}^M \right) + \sum_{i=1}^M \left(\frac{N_i^0+N_i^1}{N}\right) S(AE)_{\sigma^{(i)}}\\\\
S(EC)_\tau &= H\left( \left\{ \frac{N_i^0 + N_i^1}{N} \right\}_{i=1}^M \right) + \sum_{i=1}^M \left(\frac{N_i^0+N_i^1}{N}\right) S(E)_{\sigma^{(i)}}.
\end{align*}

From which it is clear that:
\[
S(A|EC)_\tau = \sum_{i=1}^M \left(\frac{N_i^0+N_i^1}{N}\right)S(A|E)_{\sigma^{(i)}}.
\]

We now compute $S(A|E)_{\sigma^{(i)}}$.  It is easy to see that:
\[
S(AE)_{\sigma^{(i)}} = H\left( \frac{N_i^0}{N_i^0+N_i^1}, \frac{N_i^1}{N_i^0+N_i^1}\right) = h\left(\frac{N_i^0}{N_i^0+N_i^1}\right).
\]
Thus, we need only to compute $S(E)_{\sigma^{(i)}}$.  Without loss of generality, we may write:
\begin{align*}
\ket{g_i^0} = x\ket{E} && \ket{g_i^1} = he^{i\theta}\ket{E} + d\ket{I},
\end{align*}
where $x, h, d \in \mathbb{R}$, $\braket{E|E} = \braket{I|I} = 1$, and $\braket{E|I} = 0$.  We also have the following identities:
\begin{align}
x^2 = \braket{g_i^0|g_i^0} = N_i^0\\
h^2+d^2 = \braket{g_i^1|g_i^1} = N_i^1\\\notag\\
xhe^{i\theta} = \braket{g_i^0|g_i^1} \Longrightarrow h^2 = \frac{|\braket{g_i^0|g_i^1}|^2}{N_i^0}\label{eq:ident-ev-3}
\end{align}

Writing $\sigma_E^{(i)}$ in this $\{\ket{E}, \ket{I}\}$ basis, we find:
\[
\sigma_E^{(i)} = \frac{1}{N_i^0+N_i^1}\left( \begin{array}{ccc}
x^2 + h^2 &,& he^{i\theta}d\\\\
he^{-i\theta}d &,& d^2 \end{array}\right),
\]
the eigenvalues of which are:
\[
\lambda_\pm = \frac{1}{2(N_i^0+N_i^1)} \left( N_i^0+N_i^1 \pm \sqrt{ (x^2+h^2-d^2)^2 + 4h^2d^2 } \right).
\]

Since $d^2 = N_i^1 - h^2$, and letting $\Delta = N_i^0-N_i^1$, we find:
\begin{align*}
\lambda_\pm &= \frac{1}{2(N_i^0+N_i^1)} \left( N_i^0+N_i^1 \pm \sqrt{ (\Delta+2h^2)^2 + 4h^2(N_i^1-h^2) } \right)\\
&= \frac{1}{2(N_i^0+N_i^1)} \left( N_i^0+N_i^1 \pm \sqrt{ \Delta^2 + 4h^2(\Delta+N_i^1) } \right)\\\\
&= \frac{1}{2} \pm \frac{\sqrt{ \Delta^2 + 4|\braket{g_i^0|g_i^1}|^2}}{2(N_i^0+N_i^1)},
\end{align*}
where, for the last equality, we used Identity \ref{eq:ident-ev-3}.

Thus, we have $S(E)_{\sigma^{(i)}} = H(\lambda_+, \lambda_-) = h(\lambda_+)$.  To complete this case of the proof, observe that $\lambda_+ \ge 1/2$ and that, as $\lambda_+$ increases (in particular, as $|\braket{g_i^0|g_i^1}|^2$ increases), $S(E)_{\sigma^{(i)}}$ decreases.  The goal of this theorem is to lower-bound the quantity $S(A|E) = S(AE)-S(E)$.  Thus, we upper-bound $S(E)$.

Since $|\braket{g_i^0|g_i^1}|^2 = Re^2\braket{g_i^0|g_i^1} + Im^2\braket{g_i^0|g_i^1} \ge Re^2\braket{g_i^0|g_i^1}$, if we define $\lambda_i$ as given in the theorem statement (i.e., Equation \ref{eq:thm:lambda}), then:
\[
\frac{1}{2} \le \lambda_i \le \lambda_+ \Longrightarrow S(E)_{\sigma^{(i)}} = h(\lambda_+) \le h(\lambda_i).
\]
Therefore, we may conclude:
\[
S(A|E)_{\sigma^{(i)}} \ge h\left(\frac{N_i^0}{N_i^0+N_i^1}\right) - h(\lambda_i),
\]
Setting $S_i$ equal to the above expression completes the proof for the case $N_i^j > 0$ for all $i$ and $j$.

If $N_i^0 = N_i^1 = 0$ for some $i$, it holds that $\ket{g_i^0} \equiv \ket{g_i^1} \equiv 0$.  Thus, this term does not appear in $\rho_{AE}$ and so does not contribute to the entropy computation, thus justifying setting $S_i = 0$ in this case.

If $N_i^0 > 0$ and $N_i^1 = 0$, then $\ket{g_i^1} \equiv 0$ and so does not show up in $\rho_{AE}$.  Thus, the $i$'th term in our decomposition of $\rho_{AE}$ given by Equation \ref{eq:thm:rho-exp} may be written:
\[
\sigma_{AE}^{(i)} = \kb{0}_A \otimes \frac{\kb{g_i^0}}{N_i^0}.
\]
The conditional entropy of such a system is simply $0$.  The case when $N_i^0 = 0$ and $N_i^1 > 0$ is similar, thus completing the proof.

\end{proof}

Note that the decomposition of a density matrix to the form in Equation \ref{eq:thm:1} may not be unique (indeed, one may permute the various $\ket{g_i^j}$ states).  Different decompositions lead to potentially different, though correct, lower-bounds on $S(A|E)$.  When using the above theorem to prove security of a QKD protocol, one must choose the correct ordering so as to produce the highest possible lower-bound (thus improving the key-rate expression).

\subsection{Extended B92}

In this section, we apply our method to the analysis of the Extended B92 protocol introduced in \cite{QKD-B92-extended}.  This protocol, like the standard B92 \cite{QKD-B92}, uses two non-orthogonal states to encode the raw key bits, however it extends the protocol by allowing $A$ to send other qubit states (beyond the two only allowed by standard B92) for parameter estimation purposes.  We will use our technique to derive more optimistic noise tolerances for this protocol than prior work in \cite{QKD-B92-extended}.

We begin by introducing the protocol using our terminology.  Denote by $\Psi$-B92 the protocol shown in Protocol \ref{alg:B92}.

\begin{algorithm}
\caption{$\Psi$-B92}\label{alg:B92}
\textbf{Input}:
Let $\Psi \subset Z \cup \mathcal{A}_\alpha \cup \mathcal{B}_\beta$ be the set of possible states that $A$ may send to $B$ under the restrictions that $\ket{0}, \ket{a} \in \Psi$.  If $\alpha = 0$ or $1$, then we assume a state $\ket{+} \in \Psi$.

\textbf{Quantum Communication Stage}:
The quantum communication stage of the protocol repeats the following process:
\begin{enumerate}
  \item $A$ will send a qubit state $\ket{\psi} \in \Psi$, choosing randomly according to some publicly known distribution (we assume $\ket{0}$ and $\ket{a}$ are chosen with equal probability and that all states in $\Psi$ have non-zero probability of being chosen).
  \item $B$ chooses a random basis and measures the qubit in this basis.
  \item If $A$ chose to send $\ket{0}$ or $\ket{a}$, she sets her raw key bit to be $0$ or $1$ respectively.
  \item If $B$ observes a $\ket{\bar{a}}$ or $\ket{1}$, he sets his key bit to be $0$ or $1$ respectively.
  \item $A$ and $B$ announce, over the authenticated classical channel, whether this is a \emph{successful} iteration: namely, whether $A$ choose $\ket{0}$ or $\ket{a}$ and whether $B$ observed $\ket{\bar{a}}$ or $\ket{1}$ (of course they do not disclose their actual preparations or observations).  All other iterations, along with a suitable, randomly chosen, subset of successful iterations, are used for parameter estimation as described in Section \ref{section:PE}.
\end{enumerate}
\end{algorithm}

To compute the key rate of this Extended B92 protocol, we must first describe the joint quantum system held by $A$, $B$, and $E$ conditioning on ``successful'' iterations.  Recall that $\bar{\alpha} = \sqrt{1-\alpha^2}$ (thus $\ket{a} = \alpha\ket{0} + \bar{\alpha}\ket{1}$).  We also define $P(z) = zz^*$, where $z^*$ is the conjugate transpose of $z$.  It is not difficult to show that this quantum system is:
\begin{align*}
\rho_{ABE} &= \frac{1}{N'}\left[\frac{1}{2}\kb{00}_{AB} \otimes P(\bar{\alpha}\ket{e_0} - \alpha\ket{e_1}) + \frac{1}{2}\kb{11}_{AB} \otimes P(\bar{\alpha} \ket{f_0} - \alpha\ket{f_1})\right]\\
&+ \frac{1}{N'} \left[ \frac{1}{2}\kb{01}_{AB} \otimes \kb{e_1} + \frac{1}{2}\kb{10}_{AB} \otimes \kb{f_1}\right],
\end{align*}
where we have adopted the same notation for $E$'s attack as in Equations \ref{eq:U-states} and \ref{eq:U-states-f}, and $N'$ is a normalization term to be discussed shortly.

Tracing out $B$'s system and using the fact that, from Equation \ref{eq:U-states-f2}, we find $\bar{\alpha}\ket{f_0} - \alpha\ket{f_1} = \alpha\ket{e_1} + \bar{\alpha}\ket{e_3}$, we have:
\begin{align*}
\rho_{AE} &= \frac{1}{N'}\left[ \frac{1}{2} \kb{0}_A \otimes P(\bar{\alpha}\ket{e_0} - \alpha\ket{e_1}) + \frac{1}{2}\kb{1}_A \otimes P(\alpha\ket{e_1} + \bar{\alpha}\ket{e_3})\right]\\
&+\frac{1}{N'}\left[ \frac{1}{2}\kb{0}_A \otimes \kb{e_1} + \frac{1}{2}\kb{1}_A\otimes\kb{f_1}\right].
\end{align*}

In order to apply Theorem \ref{thm:entropy}, we write the above state in the following form:
\begin{align*}
\rho_{AE} &= \frac{1}{2N'} \left( \kb{0}_A \otimes (\kb{g_1^0} + \kb{g_2^0}) + \kb{1}_A\otimes(\kb{g_1^1} + \kb{g_2^1})\right)
\end{align*}
where we defined:
\begin{align*}
&\ket{g_1^0} = \bar{\alpha}\ket{e_0} - \alpha\ket{e_1}\\
&\ket{g_1^1} = \alpha\ket{e_1} + \bar{\alpha}\ket{e_3}\\
&\ket{g_2^0} = \ket{e_1}\\
&\ket{g_2^1} = \ket{f_1}\\
&N' = \frac{1}{2}(\braket{g_1^0|g_1^0} + \braket{g_2^0|g_2^0} + \braket{g_1^1|g_1^1} + \braket{g_2^1|g_2^1}).
\end{align*}

From this, Theorem \ref{thm:entropy} may be directly applied to compute the key rate of this Extended B92 protocol, providing us with a key-rate expression for any asymmetric channel (computing $H(A|B)$ is, as stated before, trivial).  We will show an example shortly; however, writing out the algebraic expression is not enlightening as it would entail simply copying Equation \ref{eq:thm:entropy} from Theorem \ref{thm:entropy}.  Instead, we will illustrate by writing out the case for a symmetric attack (again, an enforceable assumption) allowing us to take advantage of certain simplifications in the expressions.  Furthermore, it will allow us to compare with prior work to demonstrate the advantage to using mismatched measurement bases for this protocol.  We stress, however, that symmetry is not required at this point.

Let $Q$ denote the error rate of the channel as before.  Assuming a symmetric attack (and thus $\re{0}{1} = \re{0}{2} = \re{2}{3} = \re{1}{3} = 0$; see Section \ref{section:sym}), the normalization term simplifies to:

\begin{align}
N' &= \frac{1}{2}[\bar{\alpha}^2(1-Q) + \alpha^2Q + \alpha^2Q + \bar{\alpha}^2(1-Q) + Q + Q]\notag\\
&= \bar{\alpha}^2(1-Q) + \alpha^2Q + Q\notag\\
&= 1 - \alpha^2(1 - 2Q).\label{eq:b92-N}
\end{align}
where, to derive the last equality, we used the fact that $\bar{\alpha}^2 = 1-\alpha^2$.

Let $N_C = \braket{g_1^0|g_1^0} = \braket{g_1^1|g_1^1} = \bar{\alpha}^2(1-Q) + \alpha^2Q$ (the two are equal when faced with a symmetric attack) and $N_W = \braket{g_2^0|g_2^0} + \braket{g_2^1|g_2^1} = Q$.  By Theorem \ref{thm:entropy}, we have:
\[
S(A|E) \ge \frac{N_C}{N'}(1 - h(\lambda_C)) + \frac{N_W}{N'}(1 - h(\lambda_W)),
\]
where:
\begin{align*}
\lambda_C &= \frac{1}{2} + \frac{|Re\braket{g_1^0|g_1^1}|}{2N_C}\\
\lambda_W &= \frac{1}{2} + \frac{|Re\braket{g_2^0|g_2^1}|}{2N_W}.
\end{align*}

$Re\braket{g_2^0|g_2^1}$ may be evaluated directly:
\begin{equation}
Re\braket{g_2^0|g_2^1} = Re\braket{e_1|f_1} = -\alpha^2Q + \bar{\alpha}^2\re{1}{2}.
\end{equation}
(Again we use the notation $\re{i}{j}$ to mean $Re\braket{e_i|e_j}$.)

Furthermore, we have:
\begin{equation}
Re\braket{g_1^0|g_1^1} = \bar{\alpha}^2\re{0}{3} - \alpha^2Q.
\end{equation}
Bounds on $\re{0}{3}$ and $\re{1}{2}$ are found as described in Section \ref{section:sym}.

All that remains is to compute $H(A|B)$.  But this is simply:
\begin{align}
H(A|B) &= H\left( \frac{N_C}{2N'}, \frac{N_C}{2N'}, \frac{N_W}{2N'}, \frac{N_W}{2N'} \right) - h\left( \frac{N_C+N_W}{2N'} \right)\notag\\
& = H\left( \frac{N_C}{2N'}, \frac{N_C}{2N'}, \frac{N_W}{2N'}, \frac{N_W}{2N'} \right) - 1.
\end{align}

To evaluate the key-rate assuming this symmetric attack, we must optimize over all $\re{1}{2} \in [-Q,Q]$ if $\Psi = \Psi_3$.  Otherwise, if we are using $\Psi_4$, we have exact values for these parameters and evaluate the expression directly.  Note that, while the choice of $\alpha$ does not affect parameter estimation, it does, in this case, affect the key-rate as it is used as one of the states for key distillation ($\beta$, being used only for parameter estimation, does not).  The resulting key-rate for various values of $\alpha$ and states $\Psi$ is shown in Table \ref{table:B92}.  In that same table, we compare our new tolerated bounds with those found in \cite{QKD-B92-extended}; our new key rate suffers a higher tolerated error than original work in \cite{QKD-B92-extended} (which did not utilize mismatched measurement outcomes) for all $\alpha > 0$.

\begin{table}
\centering
\begin{tabular}{r|ccccc}
$\alpha$ & 0 & $0.342$ & $0.643$ & $0.939$ & $0.985$\\
\cline{2-6}
Old Bound From \cite{QKD-B92-extended} & $11\%$ & $9.3\%$ & $5.7\%$ & $1\%$ & $0.27\%$\\
New Bound Using $\Psi_3$ & $11\%$ & $9.41\%$ & $6.19\%$ & $1.5\%$ & $0.41\%$\\
New Bound using $\Psi_4$ & $12.6\%$ & $11.06\%$ & $7.37\%$ & $1.62\%$ & $0.42\%$

\end{tabular}
\caption{Comparing our new key rate bound for the Extended B92 protocol with the one from \cite{QKD-B92-extended}.  In particular, we compare the maximally tolerated error rates of our bound with the one from \cite{QKD-B92-extended} for a depolarization channel and for various values of $\alpha = \braket{0|a}$ (where $\ket{0}$ and $\ket{a}$ are used to encode the classical value of $0$ and $1$ respectively).}\label{table:B92}
\end{table}

Before leaving this section, let us demonstrate how our key-rate expression can be used when the channel is non symmetric.  The following procedure adapts to all other protocols we consider in this paper (though the key-rate equations are, of course, different).  To demonstrate the procedure we will consider the channel statistics shown in Table \ref{table:non-sym-BB84}.  For this demonstration, we assume here that, for parameter estimation, $\alpha = \beta = 1/\sqrt{2}$.  However, we will then evaluate the key rate, assuming these statistics, if different values of $\alpha$ are used for key distillation.  That is to say, for parameter estimation, Alice will send $\ket{0}, \ket{1}, \ket{+}$, and $\ket{0_Y}$; however for key distillation she will use $\ket{0}$ and $\ket{a}$.  In practice, parameter estimation would also use $\ket{a}$; however as we wish to consider the key rate for various choices of $\alpha$, this would be too cumbersome (changing the parameter of $\alpha$ does not affect $A$ and $B$'s estimate of the various $\re{i}{j}$ used, though it does affect the observed $p_{i,j}$ - thus, to simplify our discussion, we fix the parameter estimation process).

Consider first the case where $\Psi_4$ is used for parameter estimation in which case all values in Table \ref{table:non-sym-BB84} may be observed.  Using Equations \ref{eq:p0a}, \ref{eq:p1a}, and \ref{eq:pa0}, we find:
\begin{align*}
\re{0}{1} &= p_{0,a} - \frac{1}{2}(.868 + .132) = .418 - \frac{1}{2} = -.082\\
\re{2}{3} &= p_{1,a} - \frac{1}{2}(.03+.97) = .605 - \frac{1}{2} = .105\\
\re{0}{2} &= p_{a,0} - \frac{1}{2}(.868 + .03) = .536 - .449 = .087\\
\re{1}{3} &= -\re{0}{2} = -.087
\end{align*}
(The last equality follows, as mentioned in Section \ref{section:PE}, from unitarity of $E$'s attack.)

Similarly, using Equations \ref{eq:p0b}, \ref{eq:p1b}, and \ref{eq:pb0}, we find:
\begin{align*}
\im{0}{1} &= p_{0,b} - \frac{1}{2}(.868 + .132) = .564 - \frac{1}{2} = .064\\
\im{2}{3} &= p_{1,b} - \frac{1}{2}(.03+.97) = .486 - \frac{1}{2} = -.014\\
\im{0}{2} &= \frac{1}{2}(.868+.03) - p_{b,0} = .449 - .472 = -.023\\
\im{1}{3} &= -\im{0}{2} = .023
\end{align*}

Finally, Equations \ref{eq:r03-gen}, and \ref{eq:r12-gen} give us:
\begin{align*}
\re{0}{3} &= 1 - p_{a,\bar{a}} - p_{b,\bar{b}} - \frac{1}{2}(\re{0}{1} + \im{0}{1} + \re{2}{3} + \im{2}{3})\\
&=.86 - .0365 = .8235\\
\re{1}{2} &= p_{b,\bar{b}} - p_{a,\bar{a}} + \frac{1}{2}(\im{0}{1} - \re{0}{1} + \im{2}{3} - \re{2}{3})\\
&= .036 + .0135 = .0495
\end{align*}

Using the above, we may immediately apply Theorem \ref{thm:entropy}.  To demonstrate, assume that, while for parameter estimation above, we had $\alpha = 1/\sqrt{2}$, we now use, for key distillation, $\alpha = .342$ (this may be achieved, for instance, by Alice preparing separate states $\ket{+}$ for parameter estimation, but $\ket{a}$ for key distillation - of course in practice, they would not be separate, we only use $\ket{+}$ for parameter estimation as it simplifies the arithmetic - the method to evaluate the key rate is identical regardless of the choice of this parameter) we may compute the following:
\begin{align*}
N_1^0 &= \ba^2p_{0,0} + \alpha^2p_{0,1} - 2\alpha\ba\re{0}{1} = .835\\
N_1^1 &= \alpha^2p_{0,1} + \ba^2p_{1,1} + 2\alpha\ba\re{1}{3} = .816\\
N_2^0 &= p_{0,1} = .132\\
N_2^1 &= p_{a,\ba} \text{ (From Equation \ref{eq:paa})} = .024\\
N &= N_1^0 + N_1^1 + N_2^0 + N_2^1 = 1.807\\
Re\braket{g_1^0|g_1^1} &= \alpha\ba\re{0}{1} + \ba^2\re{0}{3} - \alpha^2p_{0,1} - \alpha\ba\re{1}{3} = .713\\
Re\braket{g_2^0|g_2^1} &= \alpha\ba\re{0}{1} + \ba^2\re{1}{2} - \alpha^2p_{0,1} = \alpha\ba\re{1}{3} = .03\\\\
\lambda_1 &= .932\\
\lambda_2 &= .895\\
\end{align*}
From this, we apply Theorem \ref{thm:entropy} to find:
\[
r_{B92} \ge S(A|E) - H(A|B) \ge .598 - \left[ H\left( \frac{N_1^0}{N}, \frac{N_1^1}{N}, \frac{N_2^0}{N}, \frac{N_2^1}{N}\right) - h\left(\frac{N_1^0+N_2^1}{N}\right)\right] = .205
\]

\begin{table}
\centering
\begin{tabular}{l|cccccccc|cccc}
Statistic & $p_{0,0}$ & $p_{0,1}$ & $p_{1,0}$ & $p_{1,1}$ & $p_{a,\bar{a}}$ & $p_{0,a}$ & $p_{1,a}$ & $p_{a,0}$ & $p_{b,\bar{b}}$ & $p_{0,b}$ & $p_{1,b}$ & $p_{b,0}$\\
\hline
Observed Value & $.868$ & $.132$ & $.03$ & $.97$ & $.052$ & $.418$ & $.605$ & $.536$ & $.088$ & $.564$ & $.486$ & $.472$
\end{tabular}
\caption{An example of the observable statistics a non-symmetric channel could produce when, for parameter estimation purposes, $\alpha = \beta = 1/\sqrt{2}$.  In the text, we evaluate the key-rate of the Extended B92 protocol assuming this particular channel to demonstrate how arbitrary channels are handled.  If $\Psi_4$ is used, then all statistics shown may be observed; if $\Psi_3$ is used only those statistics to the left of the second vertical bar may be observed (i.e., those states utilizing a $\ket{b}$ cannot be observed).}\label{table:non-sym-BB84}
\end{table}

In the event only $\Psi_3$ is used, $A$ and $B$ cannot observe any statistics $p_{i,j}$ involving a $\ket{b}$ state.  In this case, they are able to determine $\re{0}{1}, \re{2}{3}, \re{0}{2}$, and $\re{1}{3}$ as above (and these numbers, for this particular channel, remain the same).  They may then use Equation \ref{eq:paa}:
\begin{align*}
p_{a,\bar{a}} &= \frac{1}{2} - \frac{1}{2}(\re{0}{2} + \re{0}{1} + \re{0}{3} + \re{1}{2} + \re{2}{3} + \re{1}{3})\\
\Rightarrow \re{0}{3} &= 1 - 2p_{a,\ba} - \re{0}{2} - \re{0}{1} - \re{2}{3} - \re{1}{3} - \re{1}{2}.
\end{align*}
(Again, since this is back to parameter estimation, and in our toy example we are using $\ket{+}$ for this, we have $\alpha = 1/\sqrt{2}$.)

In the above, the only parameter that is not known is $\re{1}{2}$.  However, the Cauchy Schwarz inequality bounds this by $|\re{1}{2}| \le \sqrt{p_{0,1}p_{1,0}}$.  Thus, one must simply optimize over all such $\re{1}{2}$ which minimizes the key-rate bound (we assume $E$ chose an attack which optimizes her information gain).  Doing so, yields a key rate, in this particular example, of $r_\text{B92} \ge .194$.

Table \ref{table:b92-asym} shows this same experiment, using data from Table \ref{table:non-sym-BB84}, for other choices of $\alpha$ (again, for key distillation - the choice for parameter estimation remains fixed at $1/\sqrt{2}$).

\begin{table}
\centering
\begin{tabular}{r|ccccc}
$\alpha$ & 0 & $0.1$ & $0.2$ & $0.342$ & $0.643$\\
\cline{2-6}
Key-Rate Using $\Psi_3$ & $.288$ & $.293$ & $.271$ & $.194$ & $0$\\
Key-Rate using $\Psi_4$ & $.292$ & $.295$ & $.275$ & $.205$ & $.012$

\end{tabular}
\caption{Computing the key rate of the Extended B92 protocol for the asymmetric channel described in Table \ref{table:non-sym-BB84} for various choices of $\alpha$ (where $\alpha$ here is only used for key distillation - for parameter estimation, we used $\ket{+}$ instead of $\ket{a}$).  Note that $\alpha=0$ is actually BB84-style encoding; also note that this is not the optimal value for this particular channel (unlike the symmetric case).}\label{table:b92-asym}
\end{table}

\subsection{An Optimized QKD Protocol}

From our analysis of the Extended B92 protocol above, it is clear that, as $\alpha$ approaches $0$ (i.e., as $\ket{a}$ approaches $\ket{1}$), the tolerated noise level of the protocol increases to the level of BB84 assuming a symmetric channel; furthermore, this produces its maximal noise tolerance.  This leads to the natural question: under what channels is the BB84 style encoding (i.e., using $\ket{i}$ to encode a key bit of $i \in \{0,1\}$) always optimal?  For the case of symmetric channels, it was shown in \cite{QKD-OptimalQKD-Symmetric} that BB84-style encoding produces an optimal result.  We will confirm this result using an alternative method.  Furthermore, we will consider arbitrary, not necessarily symmetric, channels and see that there exist such channels where BB84 would fail, yet through an optimized protocol, a key may be distilled.  While this is not a surprising result, the techniques we use to compute the key-rate of this ``optimized'' protocol, may be useful in future applications and theoretical explorations of this problem.

To investigate this, we will consider a rather general QKD protocol, where if $A$ wishes to encode a key bit of $0$, she will send a qubit $\ket{\psi_0} = \alpha_s\ket{0} + \sqrt{1-\alpha_s^2}\ket{1}$.  Otherwise, to encode a key bit of $1$, she will send the qubit $\ket{\psi_1} = \gamma_s\ket{0} + \sqrt{1-\gamma_s^2}\ket{1}$.  Here we have $\alpha_s,\gamma_s \in [-1,1]$ (the subscript ``$s$'' stands for ``send'').  $B$ will measure the incoming qubit in one of at most two bases; if he receives outcome $\ket{\phi_0} = \alpha_r\ket{0} + \sqrt{1-\alpha_r^2}\ket{1}$ he will output a raw key bit of $0$; otherwise if he measures $\ket{\phi_1} = \gamma_r\ket{0} + \sqrt{1-\gamma_r^2}\ket{1}$ he will set his key bit to be a $1$.  Any other measurement result will be considered inconclusive.  Here we have $\alpha_r,\gamma_r \in [-1,1]$ (the subscript ``$r$'' stands for ``receive''); we do not assume any necessary relationship between the values $\alpha_s, \gamma_s, \alpha_r$, and $\gamma_r$.  $B$ will announce to $A$ any rounds that were inconclusive and they will be discarded.

Additionally, for certain, randomly chosen iterations, $A$ will send a state from $\Psi_4$ and $B$ will measure in the $Z$, $\mathcal{A}$, or $\mathcal{B}$ basis.  These iterations will be used for parameter estimation using the process described in Section \ref{section:PE}.

We call the above protocol Opt-$\Pi$.

Given that our parameter estimation process returned certain statistics based on $E$'s attack, we ask: what are the optimal values for $\alpha_s, \gamma_s, \alpha_r, \gamma_r$ which maximize the key-rate expression?  Of course, the protocol we described above is, by itself, purely a theoretical one in that it assumes $A$ and $B$ know what values to use for these parameters before the protocol even begins.  Clearly, however, they can only know what optimal value to set them to after the protocol has run and parameter estimation returned various statistics on the attack used.  For our purposes, this is irrelevant as we care only to see whether BB84 style encoding always produces an optimal result.  That being said, however, one could potentially extend this to a practical prepare-and-measure QKD protocol using a method described in \cite{QKD-OptimalQKD-Symmetric}, namely as follows:

\begin{enumerate}
  \item Run the parameter estimation process, using states from $\Psi_4$, for $M$ iterations.
  \item Choose $\alpha_s, \gamma_s, \alpha_r, \gamma_r$ so as to optimize the key rate, given that $E$ uses the attack measured in step 1.
  \item Run Opt$-\Pi$ (with parameter estimation).
  \item If the statistics of the channel as determined during the execution of Opt-$\Pi$ differ from those in step one, abort.
\end{enumerate}

Step four prevents $E$ from significantly altering her attack after step one (i.e., it forces $E$ to commit to a particular class of attack strategy, even though she knows the first $M$ iterations are used for parameter estimation only).  She is, of course, free to alter her attack in such a way that does not alter any of the observed statistics.  However, such a change won't affect the key-rate expression once parameters are fixed in step two.  This is due to the fact that these parameters were chosen so as to optimize the key rate given any possible attack which conforms to the observed statistics.  Furthermore, since the process used by $A$ and $B$ in step two to actually choose parameters $\alpha_s, \gamma_s, \alpha_r, \gamma_r$ should be public knowledge, $E$ actually knows what they will chose as soon as she picks an attack strategy (i.e., she may simulate steps one and two).

There are two potential issues, however, with the above idea.  Firstly, step one, being finite, leads to the question of estimation accuracy (a point we safely ignored thus far as we have been working in the asymptotic scenario).  Thus, the statistics learned in step 1 may differ slightly from those in step 3.  This might be to the advantage of $E$ as she may alter her attack slightly after step 1.  Finally, while the above process works for collective attacks (barring the accuracy issue just mentioned), one must be careful when applying a de Finetti argument to extend this approach to general attacks (as we do in a later section of this paper).  Whether such an argument works here we leave as an open question.  However, we comment that further study of this problem along with this prepare-and-measure, optimized QKD approach, may yield interesting, and potentially practically useful, results.

Of course we will be applying Theorem \ref{thm:entropy} to compute the key-rate of this optimized protocol; before we may do so, however, we must first construct the density operator describing a single iteration of the protocol, conditioning, of course, on the event that it is used to distill a raw key.  Such an iteration begins with $A$ sending $\ket{\psi_0}$ or $\ket{\psi_1}$, choosing each with probability $1/2$.  $E$ then attacks the qubit after which it arrives at $B$'s lab.  To simplify notation, let $\beta_r = \sqrt{1-\alpha_r^2}$ and $\delta_r = \sqrt{1-\gamma_r^2}$ (similarly define $\beta_s$ and $\delta_s$).  Then, using our usual notation for $E$'s attack operator (see Equation \ref{eq:U-states}), the state, when it arrives at $B$'s lab but before his measurement, is:
\begin{align*}
&\frac{1}{2} \kb{0}_A \otimes U\kb{\psi_0}U^* + \frac{1}{2}\kb{1}_A \otimes U\kb{\psi_1}U^*\\\\
=&\frac{1}{2} \kb{0}_A \otimes P\left(\ket{0}\otimes(\overbrace{\alpha_s\ket{e_0} + \beta_s\ket{e_2}}^{\ket{f_0}}) + \ket{1}\otimes(\overbrace{\alpha_s\ket{e_1} + \beta_s\ket{e_3}}^{\ket{f_1}})\right)\\
+&\frac{1}{2} \kb{1}_A \otimes P\left(\ket{0}\otimes(\underbrace{\gamma_s\ket{e_0} + \delta_s\ket{e_2}}_{\ket{f_2}}) + \ket{1}\otimes(\underbrace{\gamma_s\ket{e_1} + \delta_s\ket{e_3}}_{\ket{f_3}})\right)\\\\
&=\frac{1}{2}\kb{0}_A \otimes P(\ket{0,f_0} + \ket{1, f_1}) + \frac{1}{2}\kb{1}_A\otimes P(\ket{0,f_2} + \ket{1,f_3}),
\end{align*}
where $P(z) = zz^*$ as before.

$B$ will now measure (perhaps choosing one of two different bases to do so) and set his raw key bit appropriately.  Conditioning on the event he receives an outcome of $\ket{\phi_0}$ or $\ket{\phi_1}$ (which, furthermore, may depend on the event that he chooses to measure in a suitable basis), the above state evolves to:

\begin{align}
\rho_{ABE} &= \frac{1}{2N'}(\kb{00}_{AB}\otimes \kb{g_1^0} + \kb{11}_{AB} \otimes \kb{g_1^1}\\
&+ \kb{01}_{AB}\otimes\kb{g_2^0} + \kb{10}_{AB}\otimes\kb{g_2^1}),\notag
\end{align}
where:
\begin{align*}
\ket{g_1^0} &= \alpha_r\ket{f_0} + \beta_r\ket{f_1}\\
\ket{g_1^1} &= \gamma_r\ket{f_2} + \delta_r\ket{f_3}\\
\ket{g_2^0} &= \gamma_r\ket{f_0} + \delta_r\ket{f_1}\\
\ket{g_2^1} &= \alpha_r\ket{f_2} + \beta_r\ket{f_3},
\end{align*}
and $N'$ is the normalization term:
\[
N' = \frac{1}{2}\sum_{i,j} N_{i,j},
\]
where $N_{i,j} = \braket{g_i^j|g_i^j}$.  (Note that $N_{i,j}$ is a parameter that $A$ and $B$ may observe.)

Because we wish to evaluate this numerically using Theorem \ref{thm:entropy}, we take the time to expand each $N_{i,j}$ in terms of inner-products of the $\ket{e_k}$ states (for which our parameter estimation process will determine bounds).  These are:

\begin{align*}
N_{1,0} &= \alpha_r^2\braket{f_0|f_0} + \beta_r^2\braket{f_1|f_1} + 2\alpha_r\beta_rRe\braket{f_0|f_1}\\
N_{1,1} &= \gamma_r^2\braket{f_2|f_2} + \delta_r^2\braket{f_3|f_3} + 2\gamma_r\delta_rRe\braket{f_2|f_3}\\
N_{2,0} &= \gamma_r^2\braket{f_0|f_0} + \delta_r^2\braket{f_1|f_1} + 2\gamma_r\delta_rRe\braket{f_0|f_1}\\
N_{2,1} &= \alpha_r^2\braket{f_2|f_2} + \beta_r^2\braket{f_3|f_3} + 2\alpha_r\beta_rRe\braket{f_2|f_3}.
\end{align*}

Expanding the above $\braket{f_i|f_j}$ states yields (we use the same notation for $p_{i,j}$ and $\re{i}{j}$ as introduced in Section \ref{section:PE}):
\begin{align*}
\braket{f_0|f_0} &= \alpha_s^2p_{0,0} + \beta_s^2p_{1,0} + 2\alpha_s\beta_s\re{0}{2}\\
\braket{f_1|f_1} &= \alpha_s^2p_{0,1} + \beta_s^2p_{1,1} + 2\alpha_s\beta_s\re{1}{3}\\
\braket{f_2|f_2} &= \gamma_s^2p_{0,0} + \delta_s^2p_{1,0} + 2\gamma_s\delta_s\re{0}{2}\\
\braket{f_3|f_3} &= \gamma_s^2p_{0,1} + \delta_s^2p_{1,1} + 2\gamma_s\delta_s\re{1}{3}\\\\
Re\braket{f_0|f_1} &= \alpha_s^2\re{0}{1} + \alpha_s\beta_s\re{0}{3} + \alpha_s\beta_s\re{1}{2} + \beta_s^2\re{2}{3}\\
Re\braket{f_2|f_3} &= \gamma_s^2\re{0}{1} + \gamma_s\delta_s\re{0}{3} + \gamma_s\delta_s\re{1}{2} + \delta_s^2\re{2}{3}
\end{align*}

To apply Theorem \ref{thm:entropy}, we will need the following identities:
\begin{align*}
Re\braket{g_1^0|g_1^1} &= \alpha_r\gamma_rRe\braket{f_0|f_2} + \alpha_r\delta_rRe\braket{f_0|f_3} + \beta_r\gamma_rRe\braket{f_1|f_2} + \beta_r\delta_rRe\braket{f_1|f_3}\\
Re\braket{g_2^0|g_2^1} &= \alpha_r\gamma_rRe\braket{f_0|f_2} + \beta_r\gamma_rRe\braket{f_0|f_3} + \alpha_r\delta_rRe\braket{f_1|f_2} + \beta_r\delta_rRe\braket{f_1|f_3}.
\end{align*}

And finally, to compute the above, we will additionally need the following:
\begin{align*}
Re\braket{f_0|f_2} &= \alpha_s\gamma_sp_{0,0} + \alpha_s\delta_s\re{0}{2} + \beta_s\gamma_s\re{0}{2} + \beta_s\delta_sp_{1,0}\\
Re\braket{f_0|f_3} &= \alpha_s\gamma_s\re{0}{1} + \alpha_s\delta_s\re{0}{3} + \beta_s\gamma_s\re{1}{2} + \beta_s\delta_s\re{2}{3}\\
Re\braket{f_1|f_2} &= \alpha_s\gamma_s\re{0}{1} + \alpha_s\delta_s\re{1}{2} + \beta_s\gamma_s\re{0}{3} + \beta_s\delta_s\re{2}{3}\\
Re\braket{f_1|f_3} &= \alpha_s\gamma_sp_{0,1} + \alpha_s\delta_s\re{1}{3} + \beta_s\gamma_s\re{1}{3} + \beta_s\delta_sp_{1,1}
\end{align*}

While the above expressions are rather tiresome to look at, they are all functions of parameters that may be directly estimated by $A$ and $B$ using the process described in Section \ref{section:PE}.  Thus, we may evaluate the above numerically with relative ease.

To investigate this protocol, we wrote a simple program, the source code of which is available online\footnote{\texttt{http://www.walterkrawec.org/math/OptimizeQKD.html}}, which, given values for $p_{i,j}$ will output the optimal values of $\alpha_\cdot$ and $\gamma_\cdot$.  For symmetric channels, it is easy to determine values for the $p_{i,j}$.  For arbitrary, asymmetric, channels, setting these values to random numbers will often lead to a physically impossible attack.  That is, we must find values which could actually arise from a unitary attack operator.  We did so by drawing random unitary operators, then simulating the parameter estimation process using such an operator, and finally finding an optimal QKD protocol.

The results of this experiment are shown in Table \ref{table:optimal} comparing the ``optimized'' protocol with the BB84 key rate using $\Psi_4$ for parameter estimation (which may be evaluated using our key rate for the B92 protocol, setting $\alpha = 0$ in the key-distillation phase).  For symmetric channels, it turns out BB84 encoding is in fact optimal, confirming, independently, a result from \cite{QKD-OptimalQKD-Symmetric}.

However, for asymmetric channels (even channels that are ``almost'' symmetric), this is of course not the case; in fact we found channels where BB84 would fail (even using $\Psi_4$ for parameter estimation), yet a key would still be distillable using an optimized encoding process.  Further study of this problem through more analytical means, may provide very interesting results and the work presented here should prove useful in such an effort.

\begin{table}
\centering
\begin{tabular}{r|cc|cccccc}
BB84 & .349 & 0.001 & .304 & .217 & .278 & 0 & 0 & 0\\
Opt-$\Pi$ & .349 & 0.001 & .411 & .39 & .373 & .192 & .117 & .05\\
\hline
$(\alpha_s, \gamma_s)$ & $(1,0)$ & $(1, 0)$ & $(-1, .32)$ & $(-.88, .14)$ & $(-.79, .8)$ & $(.92, -.04)$ & $(.11, -.94)$ & $(-.7, .77)$\\
$(\alpha_r, \gamma_r)$ & $(1,0)$ & $(1, 0)$ & $(-.13, -.97)$ & $(-.97, .62)$ & $(.5, -.64)$ & $(-.61, .98)$ & $(.98, .11)$ & $(-.46, .8)$\\
\hline
$p_{0,1}$ & .07 & .126 & .1 & .162 & .054 & .23 & .264 & .144\\
$p_{1,0}$ & .07 & .126 & .122 & .047 & .179 & .07 & .185 & .213\\
$p_{a,\bar{a}}$ & .07 & .126 & .052 & .134 &  .053 &  .094 & .162 & .144\\
$p_{b,\bar{b}}$ & .07 & .126 & .063 & .068 & .113 & .092 & .171 & .175\\
$p_{0,a}$ & .5 & .5 & .644 & .322 & .574 & .72 & .697 & .364\\
$p_{1,a}$ & .5 & .5 & .452 & .598 & .438 & .387 & .244 & .587\\
$p_{a,0}$ & .5 & .5 & .421 & .545 & .473 & .287 & .268 & .674\\
$p_{0,b}$ & .5 & .5 & .469 & .439 & .585 & .465 & .504 & .287\\
$p_{1,b}$ & .5 & .5 & .485 & .588 & .366 & .584 & .416 & .598\\
$p_{b,0}$ & .5 & .5 & .507 & .545 & .455 & .362 & .387 & .728\\

\end{tabular}
\caption{Top two rows show the key-rate of the BB84 protocol (using states from $\Psi_4$ for parameter estimation) and the key-rate of the Opt-$\Pi$ protocol for various channels (the channel statistics are shown in the lower rows).  Also shown are the optimal values for $\alpha_r, \gamma_r, \alpha_s,$ and $\gamma_s$.  Observe that the last three columns show channels where BB84 would fail ($A$ and $B$ would abort), yet a key may still be distilled by an optimal choice of encoding.  For symmetric channels (a sample of which is shown in the left-most two data columns; note we are using $\alpha = \beta = 1/\sqrt{2}$ in this table), BB84 style encoding, however, is optimal (as also discovered in \cite{QKD-OptimalQKD-Symmetric}).  The software we wrote to produce the data for this table is available online: \texttt{walterkrawec.org/math/OptimizeQKD.html}}\label{table:optimal}
\end{table}

\subsection{A Two-Way Protocol}

We now turn our attention to a QKD protocol reliant on a two-way quantum channel; that is, a channel which allows a qubit to travel from $A$ to $B$ and then back to $A$.  Such a scenario provides an attacker two opportunities to interact with the qubit.  As before, we will consider first collective attacks, an assumption which allows us to model $E$'s attack as two unitary operators $U_F$ applied in the forward direction (when the qubit travels from $A$ to $B$) and $U_R$ applied in the reverse channel (when the qubit travels from $B$ to $A$).  Both operators act on the qubit and $E$'s private quantum ancilla $\mathcal{H}_E$.  See Figure \ref{fig:two-way}.  While we assume $E$'s probe state is cleared to some zero state when the iteration starts, we do not make that assumption between rounds - i.e., when $U_F$ acts, $E$'s ancilla is in a $\ket{\chi}_E$ state; when $U_R$ acts, however, this is not the case (it acts on the same system $U_F$ operated on).  We will show how mismatched measurement bases can be applied here to gather statistics on each of the operators $U_F$, $U_R$ and the joint operator $U_RU_F$ leading to more optimistic key rate bounds than previously considered for the protocol in question.

\begin{figure}
  \centering
  \includegraphics[width=200pt]{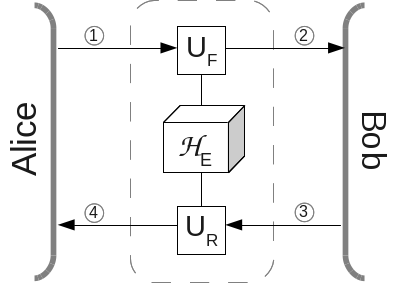}
\caption{A QKD protocol utilizing a two-way quantum channel.  $A$ begins each iteration by sending a qubit to $B$.  $E$ probes this qubit using unitary operator $U_F$, acting on the qubit and $E$'s ancilla $\mathcal{H}_E$.  Later, when the qubit returns from $B$ to $A$, $E$ attacks with a second unitary operator $U_R$ which also acts on the qubit and the same ancilla space $\mathcal{H}_E$.}\label{fig:two-way}
\end{figure}

The protocol we chose to analyze is the one described by Boyer et al., in \cite{SQKD-first}.  It is actually a \emph{semi-quantum} protocol as it makes the assumption that, while the user $A$ is \emph{fully quantum} (i.e., she may prepare and measure qubits in any basis she likes), the other user $B$ is \emph{semi-quantum} or ``classical'' in nature in that he is limited only to measuring and preparing qubits in the computational $Z$ basis.

Let us more formally describe the capabilities of the parties involved.  $A$ is permitted to send any qubit chosen from a set $\Psi$ (as in the previous section we will consider two cases: $\Psi = \Psi_3^\alpha$ and $\Psi = \Psi_4^{\alpha,\beta}$).  This qubit is sent to the ``classical'' user $B$ who may chose to do one of two things:
\begin{enumerate}
  \item \textbf{Measure and Resend}: $B$ will perform a $Z$ basis measurement on the incoming qubit resulting in outcome $\ket{r}$ for $r \in \{0,1\}$ (possibly random if $A$ sent a non-$Z$ basis state).  He will then prepare a new qubit in this state and send it to $A$; i.e., he will send $\ket{r}$ to $A$.
  \item \textbf{Reflect}: $B$ will ``reflect'' the qubit back to $A$.  That is, the qubit will pass through $B$'s lab undisturbed and return to $A$.  In this event $B$ learns nothing about the state of the qubit.
\end{enumerate}

We chose to analyze this particular protocol for several reasons.  Firstly, semi-quantum cryptography is very interesting from a theoretical stand-point as it attempts to answer the question ``how quantum does a protocol need to be to gain an advantage over a classical one?'' \cite{SQKD-first,SQKD-second}.  Thus, finding good bounds on this protocol's key rate is an important question.  Secondly, a lower bound on its key-rate was derived in \cite{SQKD-Krawec-SecurityProof} which did not make use of mismatched measurement bases.  Thus this gives us an additional comparison case showing the advantage of this method of parameter estimation.  Thirdly, it is a two-way protocol which also admits the possibility of mismatched measurements (not all two-way protocols permit such a possibility) and as such allows us to determine how mismatched measurement bases may be used to collect various statistics on the separate unitary operators under $E$'s employ: work which may extend to the analysis of other two-way QKD protocols.  Finally, it shows our technique is strong enough to provide good results even when one of the two users is limited in his capabilities (i.e., $B$ can never measure in the $\mathcal{A}$ or $\mathcal{B}$ bases).

Let us now describe the protocol we consider.  It is a generalized version of the one introduced in \cite{SQKD-first}.  In their protocol, the fully quantum user $A$ prepared qubits randomly in either the $Z$ or $X$ basis.  Here we consider the $Z$, $\mathcal{A}$, and potentially $\mathcal{B}$ basis.  Of course the classical user may only measure in the $Z$ basis.  We will denote this generalized protocol by $\Psi$-SQKD, where $\Psi$ is the set of states that $A$ is allowed to prepare.  This SQKD protocol is described in Protocol \ref{alg:SQKD}.

\begin{algorithm}
\caption{$\Psi$-SQKD}\label{alg:SQKD}
\textbf{Input}:
Let $\Psi \subset Z \cup \mathcal{A}_{\alpha} \cup \mathcal{B}_\beta$ be the set of possible states that $A$ may send to $B$ under the restrictions that $\ket{0}, \ket{1} \in \Psi$.

\textbf{Quantum Communication Stage}:
The quantum communication stage of the protocol repeats the following process:
\begin{enumerate}
  \item $A$ will send a qubit state $\ket{\psi} \in \Psi$, choosing randomly according to some publicly known distribution (we assume $\ket{0}$ and $\ket{1}$ are chosen with equal probability and that all states in $\Psi$ have non-zero probability of being chosen).

    \item $B$ will, with probability $p$ perform the \textbf{measure and resend} operation saving his measurement result.  Otherwise, he will \textbf{reflect} the qubit.
    
    \item With probability $q$, $A$ will measure the returning qubit in the $Z$ basis; otherwise she will measure in the $\mathcal{A}$ or $\mathcal{B}$ basis (assuming a state from such a basis appears in $\Psi$).
    
    \item $A$ will disclose her choice of bases and $B$ will disclose his choice of operation.  On certain, randomly chosen, iterations, $A$ and $B$ will also disclose their measurement results in order to run parameter estimation.  This disclosure is done using the authenticated classical channel.
    
    \item If this iteration is not used for parameter estimation, and if $A$ chose the $Z$ basis in steps 1 and 3, and if $B$ chose to measure and resend, they will use this iteration to contribute towards their raw key.  $B$ will use his measurement result as his key bit; $A$ will use her measurement result from Step 3 as her key bit (an alternative construction would be to use her preparation choice).
\end{enumerate}
\end{algorithm}

Note that we altered Step 3 of the original Boyer et al. \cite{SQKD-first}, protocol so as to allow the chance for mismatched measurements to occur on reflection iterations.  Furthermore, on step 5 (which is also a slight alteration from the original), the reasoning for only using those iterations where $A$ chose $Z$ basis states in both steps (as opposed to only the last) is simply to allow us to compare with prior work.  Similar computations may be performed without this stipulation.  Note that, as with BB84, one may set $p$ and $q$ arbitrarily close to one so as to improve the efficiency of the protocol \cite{QKD-BB84-Modification,SQKD-Krawec-SecurityProof}.

When $\Psi = Z \cup X$, this protocol is exactly that described in \cite{SQKD-first} (excepting for those minor alterations we mentioned above).  We will actually analyze the case when $\Psi = \Psi_3^\alpha$ and $\Psi = \Psi_4^{\alpha,\beta}$.  For the former, when $\alpha = 1/\sqrt{2}$, the protocol is actually one considered in \cite{SQKD-lessthan4} (though in that paper, the protocol was only proven \emph{robust}, a far weaker definition of security than the one we consider in this paper).  The latter case has not been considered.  Regardless, it is clear that, if we develop a lower bound on the key rate when $\Psi = \Psi_3^{\alpha}$, we will also have a lower-bound on the original SQKD protocol (i.e., we will have a lower-bound on the case when $\Psi = Z \cup X$ as dictated by the original protocol of Boyer et al.\cite{SQKD-first}).

As in the previous subsections, to compute the key-rate, we must first describe the joint quantum system held between $A$, $B$, and $E$ after one successful iteration of the protocol (where, as before, we say successful to imply it was used to contribute towards the actual raw key, and not only for parameter estimation).  Let $U_F$ be the unitary attack operator employed by $E$ in the forward channel and $U_R$ that employed by $E$ in the reverse (again, we are first assuming collective attacks).  We may also assume, without loss of generality, that at the start of the iteration $E$'s ancilla is cleared to some $\ket{\chi}_E$ state.  Thus, we may write $U_F$'s action as follows:
\begin{align*}
U_F\ket{0,\chi}_{TE} &= \ket{0,e_0} + \ket{1,e_1}\\
U_F\ket{1,\chi}_{TE} &= \ket{0,e_2} + \ket{1,e_3}.
\end{align*}
(As we did in Section \ref{section:PE}.)

We will write, without loss of generality, $U_R$'s action on states of the form $\ket{i,e_j}$ as:
\[
U_R\ket{i,e_j} = \ket{0, e_{i,j}^0} + \ket{1, e_{i,j}^1}.
\]

Conditioning on a successful iteration, $A$ will send either $\ket{0}$ or $\ket{1}$ choosing each with probability $1/2$.  $E$ will attack with operator $U_F$ and $B$ will measure in the $Z$ basis, save his result in a private register, and resend his result.  This qubit then passes through $E$ a second time.  At this point (when $E$ captures the returning qubit but before applying $U_R$), the joint system is easily found to be:

\begin{align*}
&\frac{1}{2} ( \kb{0}_B \otimes \kb{0,e_0}_{TE} + \kb{1}_B \otimes \kb{1,e_1}_{TE})\\
+&\frac{1}{2} (\kb{0}_B \otimes \kb{0,e_2}_{TE} + \kb{1}_B \otimes \kb{1,e_3}_{TE}).
\end{align*}

$E$ will now attack with her second operator $U_R$ and forward the transit qubit to $A$ who subsequently measures in the $Z$ basis.  The outcome of her measurement is saved in a private register (it is to be her raw key bit for this iteration).  The joint system is now of the form:

\begin{align*}
&\frac{1}{2} \kb{00}_{BA} \otimes \left( \kb{e_{0,0}^0} + \kb{e_{0,2}^0} \right)\\
+&\frac{1}{2} \kb{11}_{BA} \otimes \left( \kb{e_{1,3}^1} + \kb{e_{1,1}^1} \right)\\
+&\frac{1}{2} \kb{01}_{BA} \otimes \left( \kb{e_{0,0}^1} + \kb{e_{0,2}^1} \right)\\
+&\frac{1}{2} \kb{10}_{BA} \otimes \left( \kb{e_{1,3}^0} + \kb{e_{1,1}^0} \right).
\end{align*}

We will now compute $r = S(B|E) - H(B|A)$, the key rate when reverse reconciliation is used.  As described in \cite{SQKD-Krawec-SecurityProof} this seems the more natural choice for this protocol; furthermore, it will allow us to immediately compare our new key rate bound with the old one from \cite{SQKD-Krawec-SecurityProof}.

Writing $\mathbf{e_{i,j}^k}$ to mean $\kb{e_{i,j}^k}$ and tracing out $A$'s system leaves us with the state:
\begin{align}
\rho_{BE} &= \frac{1}{2}\kb{0}_B \otimes(\underbrace{\mathbf{e_{0,0}^0}}_{\mathbf{g_1^0}} + \underbrace{\mathbf{e_{0,2}^0}}_{\mathbf{g_2^0}} + \underbrace{\mathbf{e_{0,0}^1}}_{\mathbf{g_3^0}} + \underbrace{\mathbf{e_{0,2}^1}}_{\mathbf{g_4^0}})\notag\\\notag\\
&+ \frac{1}{2}\kb{1}_B \otimes(\underbrace{\mathbf{e_{1,3}^1}}_{\mathbf{g_1^1}} + \underbrace{\mathbf{e_{1,1}^1}}_{\mathbf{g_2^1}} + \underbrace{\mathbf{e_{1,3}^0}}_{\mathbf{g_3^1}} + \underbrace{\mathbf{e_{1,1}^0}}_{\mathbf{g_4^1}}).\label{eq:sqkd-state}
\end{align}
(In the above state, we also defined those states $\mathbf{g_i^j} = \kb{g_i^j}$ for use when later applying Theorem \ref{thm:entropy}.)

At this point, let us pause to determine some of the statistics which may be gathered through channel tomography.  It is clear that $A$ and $B$ may learn $\bk{i}{i}$; that is, the $Z$ basis noise in the forward channel (note that, while $B$ may measure in the $Z$ basis, he cannot measure in the $\mathcal{A}$ basis, so, for instance, he cannot estimate $p_{0,a}$).  They may also learn $\braket{e_{0,0}^i|e_{0,0}^i}, \braket{e_{0,2}^i|e_{0,2}^i}, \braket{e_{1,1}^i|e_{1,1}^i},$ and $\braket{e_{1,3}^i|e_{1,3}^i}$ for $i = 0,1$.  Indeed, consider $\braket{e_{0,0}^1|e_{0,0}^1}$.  If $A$ sends $\ket{0}$ and $B$ measures $\ket{0}$ (after $E$'s first attack), the joint state becomes:
\[
\frac{\ket{0,e_0}}{\sqrt{\bk{0}{0}}}
\]

Now, $B$ forwards the qubit to $A$, however before it arrives, $E$ attacks with $U_R$ evolving the state to:
\[
\frac{1}{\sqrt{\bk{0}{0}}} (\ket{0,e_{0,0}^0} + \ket{1, e_{0,0}^1}).
\]
Thus, the probability of $A$ measuring $\ket{1}$ is simply $\braket{e_{0,0}^1|e_{0,0}^1}/\bk{0}{0}$, thus providing $A$ and $B$ with the value $\braket{e_{0,0}^1|e_{0,0}^1}$.  Similarly, $A$ and $B$ may directly observe the other $\braket{e_{0,0}^i|e_{0,0}^i}, \braket{e_{0,2}^i|e_{0,2}^i}, \braket{e_{1,1}^i|e_{1,1}^i},$ and $\braket{e_{1,3}^i|e_{1,3}^i}$ for $i = 0,1$.

Note that, if we consider a symmetric attack (which imposes certain simplifications), we may denote by $Q_F$ the $Z$ basis error in the forward channel and $Q_R$ the $Z$ basis error in the reverse channel.  Thus:
\[
\frac{\braket{e_{0,0}^1|e_{0,0}^1}}{\bk{0}{0}} = \frac{\braket{e_{0,0}^1|e_{0,0}^1}}{1-Q_F} = Q_R \Rightarrow \braket{e_{0,0}^1|e_{0,0}^1} = Q_R(1-Q_F).
\]

Similarly, we have:
\begin{align}
(1-Q_F)(1-Q_R) &= \braket{e_{0,0}^0|e_{0,0}^0} = \braket{e_{1,3}^1|e_{1,3}^1}\label{eq:prbounds}\\
(1-Q_F)Q_R &= \braket{e_{0,0}^1|e_{0,0}^1}  = \braket{e_{1,3}^0|e_{1,3}^0}\notag\\
Q_FQ_R &= \braket{e_{1,1}^0|e_{1,1}^0}  = \braket{e_{0,2}^1|e_{0,2}^1}\notag\\
Q_F(1-Q_R) &= \braket{e_{1,1}^1|e_{1,1}^1} = \braket{e_{0,2}^0|e_{0,2}^0}\notag
\end{align}
Of course, as with the other protocols considered, we do not need any symmetry assumption - it only helps to illustrate the resulting bounds and allows us to compare with prior work.

In the following, we consider only $\ket{a} = \ket{+}$ (i.e., $\alpha = 1/\sqrt{2}$).  The case for arbitrary $\alpha$ follows exactly the same arguments, though the problem devolves into an exercise in trivial algebra.  To determine a bound on $S(B|E)$ we will need bounds on $Re\braket{g_i^0|g_i^1}$ for $i=1, 2, 3, $ and $4$.  To do so, we will utilize the mismatched measurement process described before, with certain additions appropriate for our two-way quantum channel.

Note that $A$ and $B$ may, by using $p_{a,0}$ (which we use to denote the probability that, if $A$ sends $\ket{a} = \ket{+}$ that $B$ measures $\ket{0}$), determine $\re{0}{2}$ (and subsequently $\re{1}{3}=-\re{0}{2}$) from Equation \ref{eq:pa0} (where, as before, $\re{i}{j} = Re\braket{e_i|e_j}$ and are thus statistics on $U_F$ the forward attack operator).  In particular, we have:
\[
\re{0}{2} = p_{a,0} - \frac{1}{2}(p_{0,0} + p_{1,0}) = -\re{1}{3}.
\]
Unitarity of the reverse operator $U_R$, then, forces the following relation:
\begin{align*}
\re{0}{2} &= Re\braket{e_{0,0}^0|e_{0,2}^0} + Re\braket{e_{0,0}^1|e_{0,2}^1} = p_{a,0} - \frac{1}{2}(p_{0,0} + p_{1,0})\\
\re{1}{3} &= Re\braket{e_{1,1}^0|e_{1,3}^0} + Re\braket{e_{1,1}^1|e_{1,3}^1} = \frac{1}{2}(p_{0,0} + p_{1,0}) - p_{a,0}.
\end{align*}

Now, consider the probability that $A$ measures $\ket{+}$ if she initially sent $\ket{0}$ and conditioning on the event that $B$'s measurement outcome is $\ket{0}$.  We denote this probability by $p_{0,0,a}$ and it is clearly a statistic that $A$ and $B$ may measure.  It is not difficult to compute this probability by following the evolution of the qubit in this event:
\[
\ket{0} \mapsto \ket{0,e_0}+\ket{1,e_1} \mapsto \frac{\ket{0,e_0}}{\sqrt{p_{0,0}}} \mapsto \frac{ \ket{0,e_{0,0}^0} + \ket{1,e_{0,0}^1}}{\sqrt{p_{0,0}}},
\]
from which we determine:
\begin{equation}
p_{0,0,a} = \frac{1}{2p_{0,0}}(p_{0,0} + 2Re\braket{e_{0,0}^0|e_{0,0}^1}) = \frac{1}{2} + \frac{Re\braket{e_{0,0}^0|e_{0,0}^1}}{p_{0,0}} \Rightarrow Re\braket{e_{0,0}^0|e_{0,0}^1} = p_{0,0}\left(p_{0,0,a} - \frac{1}{2}\right).
\end{equation}
This provides $A$ and $B$ with $Re\braket{e_{0,0}^0|e_{0,0}^1}$.  Similarly, we have:
\begin{align}
p_{1,0,a} &= \frac{1}{2} + \frac{Re\braket{e_{0,2}^0|e_{0,2}^1}}{p_{1,0}} \Rightarrow Re\braket{e_{0,2}^0|e_{0,2}^1} = p_{1,0}\left(p_{1,0,a} - \frac{1}{2}\right)\\
p_{0,1,a} &= \frac{1}{2} + \frac{Re\braket{e_{1,1}^0|e_{1,1}^1}}{p_{0,1}} \Rightarrow Re\braket{e_{1,1}^0|e_{1,1}^1} = p_{0,1}\left(p_{0,1,a} - \frac{1}{2}\right)\\
p_{1,1,a} &= \frac{1}{2} + \frac{Re\braket{e_{1,3}^0|e_{1,3}^1}}{p_{1,1}} \Rightarrow Re\braket{e_{1,3}^0|e_{1,3}^1} = p_{1,1}\left(p_{1,1,a} - \frac{1}{2}\right).
\end{align}
Observe that, in a symmetric attack, where $p_{i,j,a} = 1/2$, the real parts of the above listed inner-products are all zero.

Next, consider the probability that $A$ measures $\ket{+}$ if she initially sent $\ket{+}$ and conditioning on the event $B$ measures a $\ket{0}$.  We denote this probability by $p_{a,0,a}$.  In such an event it is not difficult to describe the qubit's evolution:
\begin{align*}
\ket{+} &\mapsto \frac{1}{\sqrt{2}}\ket{0}(\ket{e_0} + \ket{e_2}) + \frac{1}{\sqrt{2}}\ket{1}(\ket{e_1}+\ket{e_3}) \mapsto \frac{\ket{0}(\ket{e_0} + \ket{e_2})}{\sqrt{2p_{a,0}}}\\
&\mapsto \frac{ \ket{0}(\ket{e_{0,0}^0} + \ket{e_{0,2}^0}) + \ket{1} (\ket{e_{0,0}^1} + \ket{e_{0,2}^1})}{\sqrt{2p_{a,0}}}\\
&\mapsto \frac{1}{\sqrt{4p_{a,0}}}\ket{+}( \ket{e_{0,0}^0} + \ket{e_{0,0}^1} + \ket{e_{0,2}^0} + \ket{e_{0,2}^1}) + \frac{1}{\sqrt{4p_{a,0}}}\ket{-}(\ket{e_{0,0}^0} - \ket{e_{0,0}^1} + \ket{e_{0,2}^0} - \ket{e_{0,2}^1})
\end{align*}
From which, it is clear that:
\begin{align*}
p_{a,0,a} &= \frac{1}{4p_{a,0}}(p_{0,0} + p_{1,0} + 2Re( \braket{e_{0,0}^0|e_{0,0}^1} + \braket{e_{0,0}^0|e_{0,2}^0} + \braket{e_{0,0}^0|e_{0,2}^1} + \braket{e_{0,0}^1|e_{0,2}^0} + \braket{e_{0,0}^1|e_{0,2}^1} + \braket{e_{0,2}^0|e_{0,2}^1})).
\end{align*}

Combining this with the above work yields:
\begin{align}
Re\braket{e_{0,0}^0|e_{0,2}^1} + \braket{e_{0,0}^1|e_{0,2}^0} &= 2p_{a,0}p_{a,0,a} - \frac{1}{2}(p_{0,0}+p_{1,0}) - p_{0,0}\left(p_{0,0,a}-\frac{1}{2}\right) - p_{1,0}\left(p_{1,0,a}-\frac{1}{2}\right)\label{eq:sqkd-sum1}\\
& - p_{a,0} +\frac{1}{2}(p_{0,0}+p_{1,0}).\notag
\end{align}

Repeating the above considering $p_{a,1,a}$, we find:
\begin{align*}
p_{a,1,a} &= \frac{1}{4p_{a,1}}(p_{1,1} + p_{0,1} + 2Re( \braket{e_{1,1}^0|e_{1,1}^1} + \braket{e_{1,1}^0|e_{1,3}^0} + \braket{e_{1,1}^0|e_{1,3}^1} + \braket{e_{1,1}^1|e_{1,3}^0} + \braket{e_{1,1}^1|e_{1,3}^1} + \braket{e_{1,3}^0|e_{1,3}^1})).
\end{align*}
And so:
\begin{align}
Re\braket{e_{1,1}^0|e_{1,3}^1} + \braket{e_{1,1}^1|e_{1,3}^0} &= 2p_{a,1}p_{a,1,a} - \frac{1}{2}(p_{0,1}+p_{1,1}) - p_{0,1}\left(p_{0,1,a}-\frac{1}{2}\right) - p_{1,1}\left(p_{1,1,a}-\frac{1}{2}\right)\label{eq:sqkd-sum2}\\
& + p_{a,0} -\frac{1}{2}(p_{0,0}+p_{1,0}).\notag
\end{align}

The above expressions will become vital now as we consider $Q_A$, the error rate in the $\mathcal{A}$ basis - i.e., the probability that, if $A$ sends $\ket{+}$ and if $B$ reflects, then $A$ measures $\ket{-}$ (again, similar expressions may be derived for arbitrary $\alpha$).

Notice that, if $B$ reflects, his operation is, essentially, the identity operator.  Thus, conditioning on $B$'s choice to reflect, the two-way quantum channel becomes, in essence, a one-way channel with a qubit leaving $A$'s lab, $E$ attacking it via the unitary operator $V = U_RU_F$, and the qubit returning to $A$.  Let us now consider this operator $V$.  We may write its action on basis states as follows (again $E$'s lab is cleared to some state ``$\ket{\chi}_E$'' at the start of the iteration):
\begin{align*}
V\ket{0,\chi}_{TE} &= \ket{0,g_0} + \ket{1,g_1}\\
V\ket{1,\chi}_{TE} &= \ket{0,g_2} + \ket{1,g_3}.
\end{align*}

Due to the linearity of $U_R$ and $U_F$, these states are:
\begin{align}
\ket{g_0} &= \ket{e_{0,0}^0} + \ket{e_{1,1}^0}\label{eq:UEUFg}\\
\ket{g_1} &= \ket{e_{0,0}^1} + \ket{e_{1,1}^1}\notag\\
\ket{g_2} &= \ket{e_{0,2}^0} + \ket{e_{1,3}^0}\notag\\
\ket{g_3} &= \ket{e_{0,2}^1} + \ket{e_{1,3}^1}.\notag
\end{align}

Using the above notation, it is not difficult to show (see Section \ref{section:PE} replacing ``$e$'' states with ``$g$'' states) that the probability $Q_A$ is simply:
\begin{equation}\label{eq:sqkd-QA}
Q_A = \frac{1}{2} - \frac{1}{2}Re(\braket{g_0|g_1} + \braket{g_2|g_3} + \braket{g_0|g_3} + \braket{g_1|g_2}).
\end{equation}

The real part of $\braket{g_0|g_1}$ and $\braket{g_2|g_3}$ may be learned by using the parameter estimation method described in Section \ref{section:PE}, replacing ``$e$'' with ``$g$'' states and using iterations where $B$ reflects (thus, $A$ is performing the parameter estimation procedure with herself on the operator $V = U_RU_F$).  The critical portion of the above $Q_A$ identity is the sum of the real part of $\braket{g_0|g_3}$ with $\braket{g_1|g_2}$.  Using Equation \ref{eq:UEUFg}, we expand this sum as:
\begin{align*}
Re(\braket{g_0|g_3} + \braket{g_1|g_2}) &= Re(\overbrace{\braket{e_{0,0}^0|e_{1,3}^1}}^{E_1} + \overbrace{\braket{e_{1,1}^1|e_{0,2}^0}}^{E_2} + \overbrace{\braket{e_{0,0}^1|e_{1,3}^0}}^{E_3} + \overbrace{\braket{e_{1,1}^0|e_{0,2}^1}}^{E_4}\\
&+\braket{e_{0,0}^0|e_{0,2}^1} + \braket{e_{0,0}^1|e_{0,2}^0}\\
&+\braket{e_{1,1}^0|e_{1,3}^1} + \braket{e_{1,1}^1|e_{1,3}^0}).
\end{align*}
(Note that the inner-product denoted by $E_i$ is required to compute $\lambda_i$ for our lower bound of the conditional entropy of Equation \ref{eq:sqkd-state}; i.e., $E_i = Re\braket{g_i^0|g_i^1} = Re\braket{g_i^1|g_i^0}$.)

Notice that the only unknown quantities now are the values $E_i$ (the last four terms are found by Equations \ref{eq:sqkd-sum1} and \ref{eq:sqkd-sum2}); however we now have several restrictions on these states allowing us to easily compute our lower-bound numerically.  In particular, the Cauchy Schwarz inequality bounds $|E_i|$ while Equation \ref{eq:sqkd-QA} places a further restriction on them.  As it turns out this is sufficient to produce a very good lower-bound on this semi-quantum protocol's keyrate.

To illustrate, we will consider a symmetric channel - the computation for an asymmetric channel is similar.  In this case, we may denote by $Q$ the $Z$ basis noise in one direction of the channel.  That is, $Q$ denotes the probability of a $\ket{i}$ flipping to a $\ket{1-i}$ in the forward channel and also the same in the reverse.  Clearly, we have $Re\braket{g_0|g_1} = Re\braket{g_2|g_3} = 0$ (see Section \ref{section:PE}).  Also, Equations \ref{eq:sqkd-sum1} and \ref{eq:sqkd-sum2} are both zero (indeed, since we have $p_{0,0} = p_{1,1} = 1-Q$, $p_{0,1} = p_{1,0} = Q$, $p_{0,0,a} = p_{1,0,a} = p_{a,0} = 1/2$ and $p_{0,1,a} = p_{1,1,a} = p_{a,0} = 1/2$ this is not difficult to show).  Thus, Equation \ref{eq:sqkd-QA} becomes:
\begin{align}
Q_A &= \frac{1}{2} - \frac{1}{2}Re(E_1 + E_2 + E_3 + E_4)\notag\\
\Rightarrow E_1 &=1 - 2Q_A - E_2 - E_3 - E_4.\label{eq:sqkd-QA-bound}
\end{align}

Applying Theorem \ref{thm:entropy} to the SQKD's density operator expression in Equation \ref{eq:sqkd-state} yields:
\begin{equation}\label{eq:SQKD-SBE-Bound}
S(B|E) \ge (1-Q)^2(1 - h(\lambda_1)) + Q(1-Q)(1-h(\lambda_2)) + Q(1-Q)(1-h(\lambda_3)) + Q^2(1-h(\lambda_4)),
\end{equation}
where:
\begin{align*}
\lambda_1 &= \frac{1}{2} + \frac{|E_1|}{2(1-Q)^2}\\
\lambda_2 &= \frac{1}{2} + \frac{|E_2|}{2Q(1-Q)}\\
\lambda_3 &= \frac{1}{2} + \frac{|E_3|}{2Q(1-Q)}\\
\lambda_4 &= \frac{1}{2} + \frac{|E_4|}{2Q^2}\\
\end{align*}

To evaluate the bound on $S(B|E)$ we must simply optimize the $E_i$ so as to minimize the right-hand-side of Equation \ref{eq:SQKD-SBE-Bound} (we minimize the expression as we assume $E$ chooses an optimal attack strategy).  The values of $E_i$, however, are restricted in the following manner:
\begin{enumerate}
\item $E_1 = 1-2Q_A - E_2 - E_3 - E_4$
\item $|E_1| \le (1-Q)^2$
\item $|E_2| \le Q(1-Q)$
\item $|E_3| \le Q(1-Q)$
\item $|E_4| \le Q^2$
\end{enumerate}
(Condition 1 follows from Equation \ref{eq:sqkd-QA-bound} while Conditions 2-5 follow from the Cauchy Schwarz inequality.)

This optimization is easily performed numerically.  We consider now two cases: First both channels are independent in that $Q_{A} = 2Q(1-Q)$.  The second scenario involves the two channels being correlated in that $Q_{A} = Q$.  In the first case (the independent channel case) we discover the key-rate of the semi-quantum protocol remains positive for all $Q < 7.9\%$.  In the second scenario (the correlated case), the key-rate remains positive for all $Q < 11\%$ exactly that achieved by the fully-quantum BB84 protocol!  In fact, the key rate of the BB84 protocol when the $\mathcal{A}$ basis noise is $2Q(1-Q)$ also remains positive only up to $Q < 7.9\%$.  This is a substantial improvement over previous work as shown in Table \ref{table:sqkd}.  This is also the same noise tolerance supported by the fully quantum, LM05 \cite{QKD-TwoWay-LM05} protocol (which also requires a two-way channel) as demonstrated in \cite{QKD-twoway2} (see Figure 4 and 5 in Ref \cite{QKD-twoway2} for a graph of the LM05's key-rate assuming independent and correlated channels respectively).

\begin{table}
\centering
\begin{tabular}{r|cc}
&Independent & Correlated\\
\cline{2-3}
Old Bound From \cite{SQKD-Krawec-SecurityProof} & $4.57\%$ & $5.34\%$\\
New Bound & $7.9\%$ & $11\%$ \\

\end{tabular}
\caption{Evaluating our new key-rate bound (or, rather, the noise level where it reaches zero) for the Semi-Quantum protocol of Boyer et al. \cite{SQKD-first}.  Also comparing with the original lower-bound from \cite{SQKD-Krawec-SecurityProof} which did not utilize statistics from mismatched measurement results (note that both results are lower-bounds, so there is no contradiction).}\label{table:sqkd}
\end{table}

We have thus shown that a semi-quantum protocol, even though one of the users is limited to performing ``classical'' operations only, can in fact tolerate error rates as high as fully quantum protocols.  Since one of the original motivational reasons to study semi-quantum cryptography was to better understand how quantum a protocol needs to be in order to gain an advantage over its classical counterpart, this is a significant result.

Note that we have so far only considered the case when $A$ sends states from $\Psi_3$.  Interestingly, if we were to attempt to utilize states in $\Psi_4$ (in particular the $\ket{b}$ state), our key-rate bound did not improve using the techniques in this section.  This is most likely because we are not utilizing all possible mismatched measurement statistics for two-way channels - future work may greatly improve this result for the four-state, three-basis scenario.  We conjecture that by adding this forth state, the protocol would attain the same level of security as the six-state BB84 protocol.


\subsection{General Attacks}

In the previous sections, we considered only collective attacks.  However, all protocols considered in this paper may be made permutation invariant in the usual way \cite{QKD-survey,QKD-renner-keyrate}.  Thus, the results of \cite{QKD-general-attack,QKD-general-attack2} apply: namely, proving security against collective attacks is sufficient to show security against arbitrary, general attacks.  In the asymptotic scenario, which we considered here, the key-rate expressions will remain the same.

\section{Closing Remarks}

In this paper, we considered an Extended B92, an optimized QKD protocol, and a semi-quantum protocol.  Our new key-rate bounds for the Extended B92 protocol show it has a higher tolerance to noise than previously thought.  Similarly we derived improved key-rate bounds for the semi-quantum protocol of Boyer et al.\cite{SQKD-first}  In all cases, we did not require any symmetry assumptions (we evaluated our key-rate bounds using a symmetric attack for illustrative and comparative purposes only).

One might attempt to use this technique on other two-way protocols beyond the class of semi-quantum ones.  We did consider this - however the primary advantage of many two-way, fully quantum protocols (i.e., not semi-quantum that we considered here) is that there are no mismatched measurement basis choices.  Thus a modification to the protocols would be required - an improved key rate may be found in this case, but one must then ask what the resulting advantage would be to the modified protocol.  We leave this as an open question.

We also leave as future study improving our parameter estimation method for two-way protocols when $\Psi_4$ is used.

Further study of the use of this technique to optimize a QKD protocol for fixed channels may also provide very interesting theoretical and practical results.

Finally, and very importantly, would be to study the performance of this method in finite key settings and when imperfect parameter estimation occurs.  All results in this paper assumed $A$ and $B$ could perform enough iterations so as to derive arbitrarily precise estimates of various statistics.  In a practical setting, there will always be some error.  Taking this into account, and deriving key-rate expressions in the finite key setting is important future work.


\end{document}